\documentclass[a4paper,onecolumn,10pt]{quantumarticle}
\pdfoutput=1
\usepackage[utf8]{inputenc}
\usepackage[T1]{fontenc}
\usepackage[english]{babel}
\usepackage{amsmath}
\usepackage{amssymb}
\usepackage{amsthm}
\usepackage{commath}
\usepackage{braket}
\usepackage{multicol}
\usepackage{multirow}
\usepackage{indentfirst}
\usepackage{multirow}
\usepackage{float}
\usepackage{tikz}
\usepackage{placeins}
\usepackage{algorithm}
\usepackage{algpseudocode}
\usepackage{chngcntr}
\usepackage{subcaption}
\usepackage{pgfplots} \pgfplotsset{compat=1.18}
\pgfplotsset{
    layers/layersettings/.define layer set={
        background2,
        background,
        main,
        foreground
    }{ },
    set layers=layersettings,
}
\usepackage{pgfplotstable}
\usepackage[numbers, sort&compress]{natbib}
\usepackage{hyperref}
\usepgfplotslibrary{colormaps}
\usepackage{microtype}
\usepackage{thmtools}
\usepackage{thm-restate}
\usepackage{booktabs}

\newtheorem{theorem}{Theorem}
\newtheorem*{theorem*}{Theorem}
\newtheorem{corollary}{Corollary}[theorem]
\newtheorem{proposition}{Proposition}
\newtheorem{definition}{Definition}
\newtheorem{problem}{Problem}
\newtheorem{lemma}{Lemma}
\newtheorem*{lemma*}{Lemma}

\newcommand\blfootnote[1]{%
  \begingroup
  \renewcommand\thefootnote{}\footnote{#1}%
  \addtocounter{footnote}{-1}%
  \endgroup
}

\pgfplotstableset{
    create on use/Gamma/.style={
        create col/expr={\thisrow{circuit_id} * pi / 60}
    }
}

\pgfplotstableset{
    create on use/TimeInSec/.style={
        create col/expr={\thisrow{wall_time_ns} / 1000000000}
    }
}

\newcommand{\qkeResult}[4]{
    \begin{tikzpicture}
        \begin{axis}[
            width=#2cm, 
            height=#3cm,
            enlargelimits=false,
            axis on top,
            scale only axis,
            xticklabel pos=upper, 
            xtick pos=upper,      
            y dir=reverse,        
            xlabel={$x_j$},
            ylabel={$x_i$},
            xtick={0,1,2,3,4,5,6,7,8,9},
            ytick={0,1,2,3,4,5,6,7,8,9},
            tick style={draw=none},
            tick label style={font=\small},
            xticklabels={$x_0$, $x_1$, $x_2$, $x_3$, $x_4$, $x_5$, $x_6$, $x_7$, $x_8$, $x_9$},
            yticklabels={$x_0$, $x_1$, $x_2$, $x_3$, $x_4$, $x_5$, $x_6$, $x_7$, $x_8$, $x_9$},
            point meta min=0,
            point meta max=1,
            colormap name=viridis,
            colorbar,
            colorbar style={
                width=0.23cm,
                ytick={
                    0,
                    0.2,
                    0.4,
                    0.6,
                    0.8,
                    1
                },
                yticklabels={
                    $0$,
                    $0.2$,
                    $0.4$,
                    $0.6$,
                    $0.8$,
                    $1$
                },
                yticklabel style={
                    font=\small
                },
                ylabel={$K(x_i,x_j)$},
                ylabel style={
                    font=\small
                }
            }
        ]
            \addplot [
                matrix plot,
                mesh/cols=10,
                point meta=explicit
            ] table [
                x expr={mod(\thisrow{circuit_id}, 10)},
                y expr={floor(\thisrow{circuit_id} / 10)},
                meta=result,
                col sep=comma
            ] {#1};
        \end{axis}
        \node[
            anchor=south west,
            font=\normalsize\bfseries
        ]
        at ([xshift=-1pt,yshift=14pt]current axis.north west)
        {#4};
    \end{tikzpicture}
}

\newcommand{\qkeTruth}[4]{
    \begin{tikzpicture}
        \begin{axis}[
            width=#2cm, 
            height=#3cm,
            enlargelimits=false,
            axis on top,
            scale only axis,
            xticklabel pos=upper, 
            xtick pos=upper,      
            y dir=reverse,        
            xlabel={$x_j$},
            ylabel={$x_i$},
            xtick={0,1,2,3,4,5,6,7,8,9},
            ytick={0,1,2,3,4,5,6,7,8,9},
            tick style={draw=none},
            tick label style={font=\small},
            xticklabels={$x_0$, $x_1$, $x_2$, $x_3$, $x_4$, $x_5$, $x_6$, $x_7$, $x_8$, $x_9$},
            yticklabels={$x_0$, $x_1$, $x_2$, $x_3$, $x_4$, $x_5$, $x_6$, $x_7$, $x_8$, $x_9$},
            point meta min=0,
            point meta max=1,
            colormap name=viridis,
            colorbar,
            colorbar style={
                width=0.23cm,
                ytick={
                    0,
                    0.2,
                    0.4,
                    0.6,
                    0.8,
                    1
                },
                yticklabels={
                    $0$,
                    $0.2$,
                    $0.4$,
                    $0.6$,
                    $0.8$,
                    $1$
                },
                yticklabel style={
                    font=\small
                },
                ylabel={$K(x_i,x_j)$},
                ylabel style={
                    font=\small
                }
            }
        ]
            \addplot [
                matrix plot,
                mesh/cols=10,
                point meta=explicit
            ] table [
                x expr={mod(\thisrow{circuit_id}, 10)},
                y expr={floor(\thisrow{circuit_id} / 10)},
                meta=truth,
                col sep=comma
            ] {#1};
        \end{axis}
        \node[
            anchor=south west,
            font=\normalsize\bfseries
        ]
        at ([xshift=-1pt,yshift=14pt]current axis.north west)
        {#4};
    \end{tikzpicture}
}

\newcommand{\qkeXi}[4]{
    \begin{tikzpicture}
        \begin{axis}[
            width=#2cm, 
            height=#3cm,
            enlargelimits=false,
            axis on top,
            scale only axis,
            xticklabel pos=upper, 
            xtick pos=upper,      
            y dir=reverse,        
            xlabel={$x_j$},
            ylabel={$x_i$},
            xtick={0,1,2,3,4,5,6,7,8,9},
            ytick={0,1,2,3,4,5,6,7,8,9},
            tick style={draw=none},
            tick label style={font=\small},
            xticklabels={$x_0$, $x_1$, $x_2$, $x_3$, $x_4$, $x_5$, $x_6$, $x_7$, $x_8$, $x_9$},
            yticklabels={$x_0$, $x_1$, $x_2$, $x_3$, $x_4$, $x_5$, $x_6$, $x_7$, $x_8$, $x_9$},
            colormap name=viridis,
            colorbar,
            colorbar style={
                width=0.23cm,
                ylabel={$\xi_\star(x_i,x_j)$},
                ylabel style={
                    font=\small
                },
                every y tick scale label/.style={
                    at={(1,1)},         
                    anchor=south west,
                    inner sep=0pt,
                    xshift=-5pt,
                    yshift=2pt,
                }
            }
        ]
            \addplot [
                matrix plot,
                mesh/cols=10,
                point meta=explicit
            ] table [
                x expr={mod(\thisrow{circuit_id}, 10)},
                y expr={floor(\thisrow{circuit_id} / 10)},
                meta=xi,
                col sep=comma
            ] {#1};
        \end{axis}
        \node[
            anchor=south west,
            font=\normalsize\bfseries
        ]
        at ([xshift=-1pt,yshift=14pt]current axis.north west)
        {#4};
    \end{tikzpicture}
}

\newcommand{\qaoaXi}[4]{
    \pgfplotstableread[col sep=comma]{#1}\result
    \begin{tikzpicture}
        \begin{axis}[
            width=#2cm,
            height=#3cm,
            xlabel={$\gamma$},
            ylabel={$\xi_\star$},
            xticklabels={0,$\pi/8$,$\pi/4$,$3\pi/8$,$\pi/2$},
            xtick={0,0.3926991, 0.7853982, 1.1780972, 1.5707963},
            xmax=1.62,
            xmin=-0.05,
            scale only axis
            ]
            \addplot[blue, smooth, mark=*, mark size=1pt] table [x=Gamma,y=xi, col sep=comma] \result;
            \node[
                anchor=north west,            font=\normalsize\bfseries
            ]
            at (rel axis cs:0.025,0.975)
            {#4};
        \end{axis}
    \end{tikzpicture}
}

\newcommand{\qkestddev}[4]{
    \begin{tikzpicture}
        \begin{axis}[
            width=#2cm, 
            height=#3cm,
            enlargelimits=false,
            axis on top,
            scale only axis,
            xticklabel pos=upper, 
            xtick pos=upper,      
            y dir=reverse,        
            xlabel={$x_j$},
            ylabel={$x_i$},
            xtick={0,1,2,3,4,5,6,7,8,9},
            ytick={0,1,2,3,4,5,6,7,8,9},
            tick style={draw=none},
            tick label style={font=\small},
            xticklabels={$x_0$, $x_1$, $x_2$, $x_3$, $x_4$, $x_5$, $x_6$, $x_7$, $x_8$, $x_9$},
            yticklabels={$x_0$, $x_1$, $x_2$, $x_3$, $x_4$, $x_5$, $x_6$, $x_7$, $x_8$, $x_9$},
            colorbar,
            colormap name=viridis,            
            colorbar style={
                width=0.23cm,
                yticklabel style={
                    font=\small
                },
                ylabel={$\hat{\eta}^2(x_i, x_j)$},
                ylabel style={
                    font=\small
                },
                y tick scale label style={
                    xshift=-8pt,  
                    yshift=0pt,   
                    anchor=south west
                }
            }
        ]
        \addplot [
                matrix plot,
                mesh/cols=10,
                point meta=explicit
            ] table [
                x expr={mod(\thisrow{circuit_id}, 10)},
                y expr={floor(\thisrow{circuit_id} / 10)},
                meta=variance_hat,
                col sep=comma
            ] {#1};
        \end{axis}
        \node[
            anchor=south west,
            font=\normalsize\bfseries
        ]
        at ([xshift=-1pt,yshift=14pt]current axis.north west)
        {#4};
    \end{tikzpicture}
}

\newcommand{\graphscatterA}[3]{
    \pgfplotstableread[col sep=comma]{#1}\result
    \begin{tikzpicture}     
        \begin{axis}[
            width=#2cm,
            height=#3cm,
            scale only axis,
            ymode = log,
            log basis y = {10},
            xmode = log,
            log basis x = {10},
            ticklabel style={
                font=\small
            },
            xlabel={$\xi_\star$}, 
            ylabel={Sample count},
            legend to name=sharedlegend,
            legend columns=-1,
            legend cell align=left,
            legend image post style={
                mark size=2.2pt
            },
            legend style={
                draw=none,
                fill=none,
                font=\small,
                /tikz/every even column/.append style={
                    column sep=1.3em
                }
            }
        ]
        \addplot[
            color = blue!80!black, 
            mark=o, 
            only marks, 
            mark size=1.5pt, 
            restrict expr to domain={
                \thisrow{xi}==1 ? 0 : 1}{0.5:1.5} 
        ] 
        table[
            x=xi, y=h_old, col sep=comma
        ] \result;
        \addlegendentry{$h_{\mathrm{Hoeffding}}$}
        \addplot[
            color = red!75!yellow, mark=square*, only marks, mark size=1.5pt, restrict expr to domain={\thisrow{xi}==1 ? 0 : 1}{0.5:1.5} 
        ] 
        table[
            x=xi, y=h_new, col sep=comma
        ] \result;
        \addlegendentry{$h_{\mathrm{CARVE}}$}
 
        \addplot[
            black, thick, loosely dashed, forget plot, domain=1e3:1e4, samples=2
        ] 
        {10 * x} 
        node[
            pos=0.5, below=5pt, font=\small
        ] 
        {$h\propto\xi_\star$};
            
        \addplot[
            black, thick, loosely dashed, forget plot, domain=1e3:1e4, samples=2
        ] 
        {1000*x^2}
        node[
            pos=0.3, above=9pt, font=\small
        ] 
        {$h\propto\xi_\star^2$};

        \node[
            anchor=north west,            font=\normalsize\bfseries
        ]
        at (rel axis cs:0.025,0.975)
        {(a)};
        \end{axis}
    \end{tikzpicture}
}

\newcommand{\graphscatterB}[3]{
    \pgfplotstableread[col sep=comma]{#1}\result
    \begin{tikzpicture}
        \begin{axis}[
            width=#2cm,
            height=#3cm,
            scale only axis,
            ymode = log,
            log basis y = {10},
            xmode = log,
            log basis x = {10},
            xtick={
                3e5,
                1e6,
                3e6
            },
            xticklabels={
                $3\!\times\!10^5$,
                $10^6$,
                $3\!\times\!10^6$
            },            
            ticklabel style={
                font=\small
            },
            xlabel={$\xi_\star$}, 
            ylabel={Sample count}
        ]            
        \addplot[
            forget plot, color = blue!80!black, mark=o, only marks, mark size=1.5pt, restrict expr to domain={\thisrow{xi}==1 ? 0 : 1}{0.5:1.5} 
        ] 
        table[
            x=xi, y=h_old, col sep=comma
        ] \result;
        \addplot[
            forget plot, color = red!75!yellow, mark=square*, only marks, mark size=1.5pt, restrict expr to domain={\thisrow{xi}==1 ? 0 : 1}{0.5:1.5} 
        ] 
        table[
            x=xi, y=h_new, col sep=comma
        ] \result;
            
        \addplot[
            black, thick, loosely dashed, forget plot, domain=5e5:1e6, samples=2
        ] 
        {200 * x} 
        node[
            pos=0.5, above=3pt, font=\small
        ] 
        {$h\propto\xi_\star$};
            
        \addplot[
            black, thick, loosely dashed, forget plot, domain=5e5:1e6, samples=2
        ] 
        {20 * x^2}
        node[
            pos=0.5, below=3pt, font=\small
        ] 
        {$h\propto\xi_\star^2$};

        \node[
            anchor=north west,                  font=\normalsize\bfseries
        ]
        at (rel axis cs:0.025,0.975)
        {(b)};
        \end{axis}
    \end{tikzpicture}
}

\newcommand{\qaoaA}[3]{
    \pgfplotstableread[col sep=comma]{#1}\result
    \begin{tikzpicture}
        \begin{axis}[
            width=#2cm,
            height=#3cm,
            xlabel={$\gamma$},
            ylabel={$E(\gamma)$},
            xticklabels={0,$\pi/8$,$\pi/4$,$3\pi/8$,$\pi/2$},
            xtick={0,0.3926991, 0.7853982, 1.1780972, 1.5707963},
            xmax=1.62,
            xmin=-0.05,
            scale only axis,
            ticklabel style={
                font=\small
            },
            legend style={
                at={(0.97,0.05)},
                anchor=south east,
                draw=none,
                fill=none,
                font=\scriptsize,
                legend cell align=left
            },
            clip mode = individual
        ]
        \addplot[
            blue!80!black, 
            smooth, 
            mark=none, 
            line width = 0.9pt
        ]
        table[
            x=Gamma, y=truth, col sep=comma
        ]\result;
        \addlegendentry{Exact $E(\gamma)$}
        \addplot[
            red!75!yellow, 
            densely dashed, 
            mark=o, 
            mark size=1.8pt, 
            line width = 0.8pt,
            mark options={
                solid,
                draw=red!75!yellow,
                line width=0.6pt
            }
        ] 
        table[
            x=Gamma, y=result, col sep=comma
        ]\result;
        \addlegendentry{Estimated $E(\gamma)$}
        \node[
            anchor=north west,                  font=\normalsize\bfseries
        ]
        at (rel axis cs:0.05,0.975)
        {(a)};
        \end{axis}
    \end{tikzpicture}
}

\newcommand{\qaoaB}[3]{
    \pgfplotstableread[col sep=comma]{#1}\result
    \begin{tikzpicture}
        \begin{axis}[
            width=#2cm,
            height=#3cm,
            xlabel={$\gamma$},             
            xticklabels={0,$\pi/8$,$\pi/4$,$3\pi/8$,$\pi/2$},
            xtick={0,0.3926991, 0.7853982, 1.1780972, 1.5707963},
            xmax=1.62,
            xmin=-0.05,            
            scale only axis,
            ylabel={$\hat{\eta}^2(\gamma)$},
            ticklabel style={
                font=\small
            }
        ]
        \addplot[
            blue, 
            mark=o, 
            mark size=1.8pt, 
            line width = 0.8pt,
            mark options={
                solid,
                draw=blue!80!black,
                line width=0.6pt
            },
            restrict x to domain=0.01:1.56
        ] 
        table[
            x=Gamma, y=normalized_var_hat, col sep=comma
        ] \result;
        \node[
            anchor=north west,                  font=\normalsize\bfseries
        ]
        at (rel axis cs:0.075,0.975)
        {(b)};
        \end{axis}
    \end{tikzpicture}
}

\newcommand{\qaoaC}[3]{
    \pgfplotstableread[col sep=comma]{#1}\result
    \begin{tikzpicture}
        \begin{axis}[
            width=#2cm,
            height=#3cm,
            xlabel={$\gamma$},
            ylabel={$E(\gamma)$},
            xticklabels={0,$\pi/8$,$\pi/4$,$3\pi/8$,$\pi/2$},
            xtick={0,0.3926991, 0.7853982, 1.1780972, 1.5707963},
            xmax=1.62,
            xmin=-0.05,
            scale only axis,
            ticklabel style={
                font=\small
            }
        ]
        \addplot[
            red!75!yellow, 
            densely dashed, 
            mark=o, 
            mark size=1.8pt, 
            line width = 0.8pt,
            mark options={
                solid,
                draw=red!75!yellow,
                line width=0.6pt
            }
        ] 
        table[
            x=Gamma, y=result, col sep=comma
        ] \result;
        \node[
            anchor=north west,                  font=\normalsize\bfseries
        ]
        at (rel axis cs:0.05,0.975)
        {(c)};
        \end{axis}
    \end{tikzpicture}
}

\newcommand{\qaoaD}[3]{
    \pgfplotstableread[col sep=comma]{#1}\result
    \begin{tikzpicture}
        \begin{axis}[
            width=#2cm,
            height=#3cm,
            xlabel={$\gamma$},             
            xticklabels={0,$\pi/8$,$\pi/4$,$3\pi/8$,$\pi/2$},
            xtick={0,0.3926991, 0.7853982, 1.1780972, 1.5707963},
            xmax=1.62,
            xmin=-0.05,
            scale only axis,
            ylabel={$\hat{\eta}^2(\gamma)$},
            ticklabel style={
                font=\small
            }
        ]
        \addplot[
            blue, 
            mark=o, 
            mark size=1.8pt, 
            line width = 0.8pt,
            mark options={
                solid,
                draw=blue!80!black,
                line width=0.6pt
            },
            restrict x to domain=0.01:1.56
        ] 
        table[
            x=Gamma,
            y=normalized_var_hat,
            col sep=comma
        ] \result;
        \node[
            anchor=north west,                  font=\normalsize\bfseries
        ]
        at (rel axis cs:0.075,0.975)
        {(d)};
        \end{axis}
    \end{tikzpicture}
}

\newcommand{\qkeA}[3]{
    \begin{tikzpicture}
        \begin{axis}[
            width=#2cm, 
            height=#3cm,
            enlargelimits=false,
            scale only axis,
            axis on top,
            xticklabel pos=upper, 
            xtick pos=upper,      
            y dir=reverse,        
            xlabel={$x_j$},
            ylabel={$x_i$},
            xtick={0,1,2,3,4,5,6,7,8,9},
            ytick={0,1,2,3,4,5,6,7,8,9},
            tick style={draw=none},
            tick label style={font=\small},
            xticklabels={$x_0$, $x_1$, $x_2$, $x_3$, $x_4$, $x_5$, $x_6$, $x_7$, $x_8$, $x_9$},
            yticklabels={$x_0$, $x_1$, $x_2$, $x_3$, $x_4$, $x_5$, $x_6$, $x_7$, $x_8$, $x_9$},
            point meta min=0,
            point meta max=1,
            colormap name=viridis,
        ]
            \addplot [
                matrix plot,
                mesh/cols=10,
                point meta=explicit
            ] table [
                x expr={mod(\thisrow{circuit_id}, 10)},
                y expr={floor(\thisrow{circuit_id} / 10)},
                meta=result,
                col sep=comma
            ] {#1};
        \end{axis}
        \node[
            anchor=south west,
            font=\normalsize\bfseries
        ]
        at ([xshift=-1pt,yshift=14pt]current axis.north west)
        {(a)};
    \end{tikzpicture}
}

\newcommand{\qkeB}[3]{
    \begin{tikzpicture}
        \begin{axis}[
            width=#2cm, 
            height=#3cm,
            enlargelimits=false,
            axis on top,
            scale only axis,
            xticklabel pos=upper, 
            xtick pos=upper,      
            y dir=reverse,        
            xlabel={$x_j$},
            ylabel={$x_i$},
            xtick={0,1,2,3,4,5,6,7,8,9},
            ytick={0,1,2,3,4,5,6,7,8,9},
            tick style={draw=none},
            tick label style={font=\small},
            xticklabels={$x_0$, $x_1$, $x_2$, $x_3$, $x_4$, $x_5$, $x_6$, $x_7$, $x_8$, $x_9$},
            yticklabels={$x_0$, $x_1$, $x_2$, $x_3$, $x_4$, $x_5$, $x_6$, $x_7$, $x_8$, $x_9$},
            point meta min=0,
            point meta max=1,
            colormap name=viridis,
            colorbar,
            colorbar style={
                width=0.23cm,
                ytick={
                    0,
                    0.2,
                    0.4,
                    0.6,
                    0.8,
                    1
                },
                yticklabels={
                    $0$,
                    $0.2$,
                    $0.4$,
                    $0.6$,
                    $0.8$,
                    $1$
                },
                yticklabel style={
                    font=\small
                },
                ylabel={$K(x_i,x_j)$},
                ylabel style={
                    font=\small
                }
            }
        ]
            \addplot [
                matrix plot,
                mesh/cols=10,
                point meta=explicit
            ] table [
                x expr={mod(\thisrow{circuit_id}, 10)},
                y expr={floor(\thisrow{circuit_id} / 10)},
                meta=truth,
                col sep=comma
            ] {#1};
        \end{axis}
        \node[
            anchor=south west,
            font=\normalsize\bfseries
        ]
        at ([xshift=-1pt,yshift=14pt]current axis.north west)
        {(b)};
    \end{tikzpicture}
}

\newcommand{\graphscatterC}[4]{
    \pgfplotstableread[col sep=comma]{#1}\result
    \begin{tikzpicture}
        \begin{axis}[
            width=#2cm,
            height=#3cm,
            scale only axis,
            ymode = log,
            log basis y = {10},
            xmode = log,
            log basis x = {10},
            xlabel={$\xi_\star$}, ylabel={Sample count},
            legend style={
                at={(0.4,0.8 )}, 
                anchor=east, 
                draw=none,
                fill=none,
                font=\small,
                /tikz/every even column/.append style={
                    column sep=1.3em
                }
            }]
            \addplot[
                color = blue!80!black, 
                mark=o, 
                only marks, 
                mark size=1.5pt, 
                restrict expr to domain={
                    \thisrow{xi}==1 ? 0 : 1}{0.5:1.5} 
            ] 
            table[
                x=xi, y=h_old, col sep=comma
            ] \result;
            \addlegendentry{$h_{\mathrm{Hoeffding}}$}
            \addplot[
                color = red!75!yellow, mark=square*, only marks, mark size=1.5pt, restrict expr to domain={\thisrow{xi}==1 ? 0 : 1}{0.5:1.5} 
            ] 
            table[
                x=xi, y=h_new, col sep=comma
            ] \result;
            \addlegendentry{$h_{\mathrm{CARVE}}$}
            
            \addplot[
                black, thick, loosely dashed, forget plot, domain=4:13, samples=2
            ] 
            {30 * x}
            node[
                pos=0.5, below = 5pt, font=\small
            ] 
            {$h\propto\xi_\star$};
            \addplot[
                black, thick, loosely dashed, forget plot, domain=4:13, samples=2
            ] 
            {170 * x * x}
            node[
                pos=0.5, above = 7pt, font=\small
            ] 
            {$h\propto\xi_\star^2$};
            
            \legend{$h_{\mathrm{Hoeffding}}$, $h_{\mathrm{CARVE}}$}
        \end{axis}
    \end{tikzpicture}
    \caption{#4}
}

\newcommand{\graphscatterD}[3]{
    \pgfplotstableread[col sep=comma]{#1}\result
    \begin{tikzpicture}     
        \begin{axis}[
            width=#2cm,
            height=#3cm,
            scale only axis,
            ymode = log,
            log basis y = {10},
            xmode = log,
            log basis x = {10},
            ticklabel style={
                font=\small
            },
            xlabel={$\xi_\star$}, 
            ylabel={Sample count},
            legend to name=sharedlegend2,
            legend columns=-1,
            legend cell align=left,
            legend image post style={
                mark size=2.2pt
            },
            legend style={
                draw=none,
                fill=none,
                font=\small,
                /tikz/every even column/.append style={
                    column sep=1.3em
                }
            }
        ]
        \addplot[
            color = blue!80!black, 
            mark=o, 
            only marks, 
            mark size=1.5pt, 
            restrict expr to domain={
                \thisrow{xi}==1 ? 0 : 1}{0.5:1.5} 
        ] 
        table[
            x=xi, y=h_old, col sep=comma
        ] \result;
        \addlegendentry{$h_{\mathrm{Hoeffding}}$}
        \addplot[
            color = red!75!yellow, mark=square*, only marks, mark size=1.5pt, restrict expr to domain={\thisrow{xi}==1 ? 0 : 1}{0.5:1.5} 
        ] 
        table[
            x=xi, y=h_new, col sep=comma
        ] \result;
        \addlegendentry{$h_{\mathrm{CARVE}}$}
 
        \addplot[
            black, thick, loosely dashed, forget plot, domain=1.3:1.45, samples=2
        ] 
        {70 * x}
        node[
            pos=0.7, below = 3pt, font=\small
        ] 
        {$h\propto\xi_\star$};

        \addplot[
            black, thick, loosely dashed, forget plot, domain=1.3:1.45, samples=2
        ] 
        {130 * x * x}
        node[
            pos=0.3, above = 5pt, font=\small
        ] 
        {$h\propto\xi_\star^2$};

        \node[
            anchor=north west,            font=\normalsize\bfseries
        ]
        at (rel axis cs:0.025,0.975)
        {(a)};
        \end{axis}
    \end{tikzpicture}
}

\newcommand{\graphscatterE}[3]{
    \pgfplotstableread[col sep=comma]{#1}\result
    \begin{tikzpicture}
        \begin{axis}[
            width=#2cm,
            height=#3cm,
            scale only axis,
            ymode = log,
            log basis y = {10},
            xmode = log,
            log basis x = {10},
            ticklabel style={
                font=\small
            },
            xlabel={$\xi_\star$}, 
            ylabel={Sample count}
        ]            
        \addplot[
            forget plot, color = blue!80!black, mark=o, only marks, mark size=1.5pt, restrict expr to domain={\thisrow{xi}==1 ? 0 : 1}{0.5:1.5} 
        ] 
        table[
            x=xi, y=h_old, col sep=comma
        ] \result;
        \addplot[
            forget plot, color = red!75!yellow, mark=square*, only marks, mark size=1.5pt, restrict expr to domain={\thisrow{xi}==1 ? 0 : 1}{0.5:1.5} 
        ] 
        table[
            x=xi, y=h_new, col sep=comma
        ] \result;
            
        \addplot[
            black, thick, loosely dashed, forget plot, domain=2.5:3, samples=2
        ] 
        {100 * x}
        node[
            pos=0.5, above = 5pt, font=\small
        ] 
        {$h\propto\xi_\star$};
            
        \addplot[
            black, thick, loosely dashed, forget plot, domain=2.5:3, samples=2
        ] 
        {150 * x * x}
        node[
            pos=0.5, above = 5pt, font=\small
        ] 
        {$h\propto\xi_\star^2$};

        \node[
            anchor=north west,                  font=\normalsize\bfseries
        ]
        at (rel axis cs:0.025,0.975)
        {(b)};
        \end{axis}
    \end{tikzpicture}
}

\DeclareMathOperator{\Tr}{Tr}
\DeclareMathOperator{\poly}{poly}

\title{Efficient expectation value estimation for quantum circuits via extended stabilizer frameworks and adaptive variance estimation}
\author{Yunseo~Hwang\texorpdfstring{${}^*$}{*}}
\affiliation{Department of Computer Science and Engineering, Seoul National University, Seoul 08826, Republic of Korea}
\affiliation{NextQuantum Center, Seoul National University, Seoul 08826, Republic of Korea}
\author{Giwon~Song\texorpdfstring{${}^*$}{*}}
\affiliation{Department of Electrical and Computer Engineering, Seoul National University, Seoul 08826, Republic of Korea}
\author{Kyoung~Keun~Park}
\affiliation{Department of Computer Science and Engineering, Seoul National University, Seoul 08826, Republic of Korea}
\affiliation{NextQuantum Center, Seoul National University, Seoul 08826, Republic of Korea}
\author{Hyukjoon~Kwon\texorpdfstring{${}^\dagger$}{\dag}}
\affiliation{School of Computational Sciences, Korea Institute for Advanced Study, Seoul 02455, Republic of Korea}
\affiliation{NextQuantum Center, Seoul National University, Seoul 08826, Republic of Korea}
\author{Taehyun~Kim\texorpdfstring{${}^\dagger$}{\dag}}
\affiliation{Department of Computer Science and Engineering, Seoul National University, Seoul 08826, Republic of Korea}
\affiliation{NextQuantum Center, Seoul National University, Seoul 08826, Republic of Korea}
\affiliation{Institute of Applied Physics, Seoul National University, Seoul 08826, Republic of Korea}
\affiliation{Institute of Computer Technology, Seoul National University, Seoul 08826, Republic of Korea}
\affiliation{Automation and Systems Research Institute, Seoul National University, Seoul 08826, Republic of Korea}

\begin{document}
\blfootnote{\texorpdfstring{${}^*$}{*}These authors contributed equally to this work.}
\blfootnote{\texorpdfstring{${}^\dagger$}{\dag}Corresponding authors: hjkwon@kias.re.kr, taehyun@snu.ac.kr}

\maketitle
\begin{abstract}
    Classical simulability of quantum circuits plays a central role in characterizing and quantifying quantum computational advantage. Typically, the presence of non-Clifford gates introduces an exponential runtime overhead for classically estimating expectation values of observables. In this work, we introduce an efficient simulation framework that integrates the sum-over-Clifford method with quasi-probability-based stabilizer simulation to directly estimate expectation values. To optimize efficiency, we incorporate a circuit-adaptive reduction via variance estimation protocol, which terminates simulations once an empirical confidence bound is sufficient to guarantee the prescribed error tolerance. Our method circumvents overly conservative Hoeffding's inequality, thereby reducing the dependence on the gate-wise stabilizer extent from quadratic to linear in the low-variance regime, while maintaining the same space complexity as previous approaches. We explicitly prove the sample-complexity reduction for random quantum circuits and $T$-doped Clifford circuits and also empirically validate these advantages in the quantum approximate optimization algorithm and quantum kernel method.
\end{abstract}

\section{Introduction}
Continued advances in classical simulation algorithms are essential for benchmarking quantum hardware \cite{MPS:2003, bravyi:2016} and for understanding the boundaries of quantum computational advantage \cite{XU20254104}. In applications such as quantum chemistry and machine learning, a central computational task is the efficient estimation of observable expectation values \cite{cerezo2021variational, peruzzo2014variational, mcclean2016theory, zhao2020measurement, huggins2021efficient, temme2017error, kandala2017hardware, farhi2014quantum, biamonte2017quantum, begusic:2024, kim:2023, QSVM, NISQalgo}. Consequently, developing classical methods that estimate observables efficiently remains an important direction in quantum computing research. A variety of techniques have been developed for efficient classical simulation of restricted classes of quantum circuits \cite{gottesman1997stabilizer, gottesman:1998, valiant2002quantum, terhal2002classical, jozsa2008matchgates, shi2006classical}. Inspired by these techniques, classically demanding but practically tractable algorithms are developed \cite{koukoulekidis2022faster, ipek2026phase, wang2022possibilistic, cichy2025classical, peres2024non, hahn2025bridging,  hahn2025classical, pashayan2022fast, hamaguchi2025faster, burgholzer2021random, dias2026optimal, yashin2025further, hakkaku2021comparative, reardon2024improved, hakkaku2021sampling, piveteau2025simulating,ballarin2025optimal, yuan2024virtual, huang2025nonstabilizerness, haug2023stabilizer, zhang2026enhancing, paviglianiti2025estimating}, and they are used in quantifying quantum circuits \cite{saxena2022quantifying}, estimating physical properties \cite{hakkaku2022quantifying, wang2023universal, marshall2023simulation}, and executing quantum-classical hybrid algorithms \cite{XU20254104, miniskar2026q, jing2025circuit, luthra2025unlocking, siri2025deep, saxena2024error}. 

For near-term and fault-tolerant quantum computing, circuits that are predominantly Clifford but contain a small number of non-Clifford operations (or ``magic'') form a major region of simulation \cite{Howard_2017}. While Clifford operations are classically efficiently simulable \cite{gottesman1997stabilizer, gottesman:1998, aaronson:2004}, non-Clifford gates such as $T$ gates or $R_Z(\theta)$ rotations introduce an exponential runtime overhead. To analyze these circuits, methods whose simulation costs are governed by magic monotones rather than entanglement structures---complementing matrix product state approaches \cite{MPS:2003, MPS:2023, aziz2026classicalsimulationslowmagic, Howard_2017, Seddon_2019}---have been developed \cite{bravyi:2016, bravyi:2019, SP:2017, DFS:2020, DS:2022, pashayan:2015}. Within these methods, stabilizer-rank based algorithms and quasi-probability based algorithms are prominent. While stabilizer-rank methods achieve a linear dependence on the magic monotone via fast norm estimation, quasi-probabilistic methods remain bottlenecked by quadratic overheads \cite{DFS:2020}. This is because the required number of samples is traditionally dictated by worst-case bounds, such as Hoeffding's inequality \cite{hoeffding:1963}. These worst-case bounds fail to capture the low-variance reality of many practical circuits, artificially forcing simulations into a massive number of Monte Carlo trajectories and stalling scalability.

In this work, we introduce a circuit-adaptive reduction via variance estimation~(CARVE) protocol, which is a generalized tool to overcome sampling bottlenecks in broader quasi-probabilistic propagation algorithms. By dynamically monitoring empirical variance rather than relying on overly conservative Hoeffding's inequality, our two-stage evaluation protocol circumvents traditional limitations, terminating simulations once an empirical confidence bound is sufficient to guarantee the prescribed error tolerance and confidence level. We instantiate this generalized methodology in extended stabilizer simulation for estimating nonstabilizer circuit expectations~(ESSENCE): our classical simulation algorithm that adapts the quasi-probabilistic version of extended stabilizer simulation~(ESS) framework and its sum-over-Clifford~(SoC) representation \cite{bravyi:2016, bravyi:2019} to directly sample stabilizer branches and evaluate corresponding cross-terms. This avoids massive channel-decomposition searches required by general dyadic-channel approaches \cite{DS:2022} and maintains a highly efficient space complexity of $\mathcal{O}(n^2)$. CARVE then acts as the variance-adaptive engine. Our variance-suppression analysis mathematically demonstrates that low-variance cases are representative of practical circuits, rather than rare corner cases. 

Our approach reduces the runtime dependence on the magic monotone from a quadratic bound to a linear one in the low-variance regime, as summarized in Table~\ref{tab:complexity}. We empirically validate these theoretical advantages, and the broader utility of our generalized sampling approach, through simulations on the quantum approximate optimization algorithm~(QAOA) ansatz for the MaxE3Lin2 problem \cite{farhi2014quantum} and quantum kernel circuits used in quantum support vector machines~(QSVM) \cite{QSVM}.

\begin{table*}[ht]
    \centering
    \caption{Classical resource requirements for estimating the expectation value of an $n$-qubit quantum circuit consisting $w$ gates with gate-wise stabilizer extent $\xi_\star$ and a Pauli-sparse observable or stabilizer projector, within error $\epsilon\|O\|$ and confidence $1-\delta$.}
    \label{tab:complexity}
    \vspace{0.5em} 
    
    \begin{tabular}{llcc} 
        \toprule 
        \multicolumn{2}{l}{Method} & Time complexity & Space complexity \\
        \midrule 
        
        \multirow{2}{*}{ESSENCE + CARVE} 
        & worst case & $\mathcal{O}\left(\xi_\star^2\epsilon^{-2}\log(\delta^{-1})w\poly(n)\right)$ & $\mathcal{O}(n^2)$ \\
        & best case & $\mathcal{O}\left(\xi_\star\epsilon^{-1}\log(\delta^{-1})w\poly(n)\right)$ & $\mathcal{O}(n^2)$ \\
        
        \bottomrule 
    \end{tabular}
\end{table*}

\section{ESSENCE Framework}
\subsection{Problem Definition}
For algorithms that estimate expectation values through sampling, asymptotic convergence alone is not sufficient for certifying their outputs. We require a finite-sample guarantee: for prescribed error tolerance $\epsilon$ and failure probability $\delta$, the estimator should return an estimate that is within $\epsilon$ of the true expectation value with probability at least $1-\delta$. Such a formulation makes the required sample complexity explicit through concentration bounds and gives a predictable runtime guarantee. Thus, we define the task,  the expectation value estimation with parameters $(\epsilon, \delta)$~(EVE$(\epsilon, \delta)$), as follows.

\begin{problem}[EVE$(\epsilon, \delta)$] \label{problem:eve}
    Let $\epsilon > 0, 0 < \delta < 1$. Given an $n$-qubit unitary circuit $C$ and a Hermitian observable $O$, construct a randomized classical estimator $\widehat{\mu}_{C, O}\in\mathbb{R}$ of $\mu_{C,O} = \braket{0^{\otimes n}|C^\dagger O C|0^{\otimes n}}$ such that
    \begin{equation*}
        \mathbb{P}\left[\abs{\widehat{\mu}_{C,O} - \mu_{C, O}} \ge \epsilon \|O\|\right] \le \delta,
    \end{equation*}
    where $\|O\|$ denotes the operator norm of $O$, and the probability is taken over the internal randomness of the estimator.
\end{problem}
In the baseline ESSENCE analysis below, we first derive a worst-case sufficient number of stabilizer-pair samples, governed by the stabilizer extent $\xi$ of ESS \cite{bravyi:2019}; CARVE later refines this conservative budget using an independent variance-estimation stage, all while preserving the strict $(\epsilon,\delta)$ guarantee.

\subsection{Stabilizer Extents} \label{sec:stabilizer-extents}

To quantify the classical simulation cost of non-Clifford operations within the ESSENCE framework, we define the stabilizer extent for quantum states, unitary operations, and gate decompositions. Let $\text{STAB}_n$ denote the set of $n$-qubit pure stabilizer states.

\begin{definition}[State Stabilizer Extent \cite{bravyi:2019}]
    The stabilizer extent $\xi(\ket{\psi})$ of a pure state $\ket{\psi} \in \mathbb{C}^{2^n}$ is defined as the minimum squared $\ell_1$-norm of the coefficients over all linear decompositions into stabilizer states:
    \begin{equation}
        \xi(\ket{\psi}) = \min \left\{ \left( \sum_{k} |c_k| \right)^2 \;\middle|\; \ket{\psi} = \sum_{k} c_k \ket{\phi_k}, \;\; \ket{\phi_k} \in \mathrm{STAB}_n \right\}.
    \end{equation}
\end{definition}

\begin{definition}[Global Stabilizer Extent \cite{bravyi:2019}]
    The global stabilizer extent $\xi(C)$ of an $n$-qubit unitary operation or circuit $C$ is defined as the minimum squared $\ell_1$-norm of the coefficients over all linear decompositions into Clifford unitaries:
    \begin{equation}
        \xi(C) = \min \left\{ \left( \sum_{k} |d_k| \right)^2 \;\middle|\; C = \sum_{k} d_k V_k, \;\; V_k \text{ is Clifford for all }k\right\}.
    \end{equation}
\end{definition}

The quantity $\xi(C)$ represents the theoretical lower bound on the simulation cost of ESSENCE given an optimal global Clifford decomposition. However, computing $\xi(C)$ for arbitrary multi-qubit circuits is computationally intractable.

\begin{definition}[Gate-Wise Stabilizer Extent]
    Let a circuit $C = U^{(w)} \cdots U^{(2)} U^{(1)}$ be specified as a sequence of $w$ unitaries $U^{(r)}$, where each $U^{(r)}$ admits a known optimal SoC decomposition $U^{(r)} = \sum_{j} c_{rj} V_{rj}$ into Clifford unitaries $V_{rj}$ that achieves its stabilizer extent $\xi(U^{(r)})$. The gate-wise stabilizer extent $\xi_\star$ is defined as the product of the stabilizer extents of the individual gates:
    \begin{equation}
        \xi_\star(C) = \prod_{r=1}^{w} \xi\left(U^{(r)}\right) = \prod_{r=1}^{w} \left( \sum_{j} |c_{rj}| \right)^2.
    \end{equation}
\end{definition}

Each computational trajectory $i=(j_1, \cdots, j_w)$ through the decomposition tree corresponds to selecting a specific Clifford gate sequence $V_{1j_1}, \cdots, V_{wj_w}$, producing a total trajectory coefficient $\alpha_i$ and an output state $\ket{\phi_i} \in \text{STAB}_n$ as
\begin{equation}
    \alpha_i = \prod_{r=1}^w c_{rj_r}, \qquad \ket{\phi_i} = \left(\prod_{r=1}^w V_{rj_r}\right)\ket{0^{\otimes n}}.
\end{equation}
By construction, the output state expands as $\ket{\psi} = C\ket{0^{\otimes n}} = \sum_i \alpha_i \ket{\phi_i}$, and the sum of coefficient magnitudes across all trajectories satisfies
\begin{equation}
    \|\boldsymbol{\alpha}\|_1^2 = \left( \sum_i |\alpha_i| \right)^2 = \xi_\star.
\end{equation}

It is important to note that the partitioning of a given circuit into unitary operations $U^{(r)}$ is flexible, and different decompositions of the same circuit can yield different values of $\xi_\star(C)$. A single $U^{(r)}$ need not correspond to an individual physical gate; it can represent a compiled block of multiple physical gates, effectively reducing the overall sequence length $w$. In the limiting case where the entire circuit is treated as a single operation ($w=1$), the gate-wise extent recovers the global extent: $\xi_\star(C) = \xi(C)$.

\subsection{Sampling Inner Products of Stabilizer States} \label{sec:sample-value}

Using the gate-wise trajectory parameters established in Section~\ref{sec:stabilizer-extents}, the sampling probability distribution over trajectories and the normalized phase-adjusted stabilizer states are defined \cite{bravyi:2019} as
\begin{equation}\label{eqn:sample-trajectory}
    p_i = \frac{|\alpha_i|}{\sqrt{\xi_\star}}, \qquad \ket{\widetilde{\phi}_i} = \frac{\alpha_i}{|\alpha_i|} \ket{\phi_i}.
\end{equation}
Here, we estimate the expectation value of a Hermitian observable $O$ by
\begin{equation}\label{eqn:exp-prob-sel}
    \mu_{C,O} = \braket{\psi|O|\psi} = \sum_{i,j} \alpha_j^*\alpha_i \braket{\phi_j|O|\phi_i} = \xi_\star \sum_{i,j} p_i p_j \braket{\widetilde{\phi}_{j}|O|\widetilde{\phi}_{i}}.
\end{equation}

Let $\ket{\varphi_1}$ and $\ket{\varphi_2}$ be two independent samples from the distribution $\{p_i,\ket{\widetilde{\phi}_i}\}$. Because $O$ is Hermitian, $\mu_{C,O}$ is real and the random variable $\xi_\star \mathfrak{Re}\braket{\varphi_1|O|\varphi_2}$ is an unbiased estimator of $\mu_{C,O}$.
We define the procedure \textsc{SampleValue}$(C, O)$ which returns the value $\mathfrak{Re}\braket{\varphi_1|O|\varphi_2}$ according to the distribution.
ESSENCE therefore estimates the expectation value by averaging $X_1, X_2, \cdots, X_h$ resulting from $h$ independent iterations of $\textsc{SampleValue}(C,O)$:
\begin{equation}
    \widehat{\mu}_{C,O} =
    \frac{\xi_\star}{h} \sum_{\ell=1}^{h} X_\ell.
\end{equation}
We can obtain the required number of samples $h = 2\xi_\star^2 \epsilon^{-2} \log(2\delta^{-1})$ by the Hoeffding's inequality \cite{hoeffding:1963}, as $X_\ell$ is bound by $[-\|O\|, \|O\|]$. This procedure is summarized in Algorithm~\ref{alg:ESSENCE}.

Conceptually, as detailed in Appendix~\ref{sec:dcs-equivalence}, this sampling procedure is equivalent to a dyadic channel simulator~(DCS) \cite{DS:2022} process running with a specific decomposition over dyadic stabilizer-preserving maps. By leveraging SoC, ESSENCE sidesteps the basis optimization overhead that traditionally limits DCS.

\begin{algorithm}[H]
    \caption{Solving \textsc{EVE}($\epsilon, \delta$) using ESSENCE with Hoeffding bound}
    \label{alg:ESSENCE}
    \begin{algorithmic}
        \Procedure{ExpVal}{$C, O$}
        \State $a \gets 0, \quad \xi_\star \gets \textsc{GetXi}(C), \quad h \gets 2\xi_\star^2\epsilon^{-2}\log(2\delta^{-1})$
        \For {$j=1$ to $h$}
        \State $a \gets a + \xi_\star\cdot \textsc{SampleValue}(C,O)$
        \EndFor
        \State \Return $a/h$
        \EndProcedure
    \end{algorithmic}
\end{algorithm}

The $\textsc{SampleValue}(C,O)$ subroutine requires sampling the cross-term $\langle\varphi_1|O|\varphi_2\rangle$ of a pair of stabilizer states. 
The calculation of each term proceeds in two steps. First, we invoke a subroutine that samples stabilizer states $\ket{\varphi_1}$ and $\ket{\varphi_2}$ independently over the trajectories of Eq.~\eqref{eqn:sample-trajectory}. Second, we evaluate the value $\braket{\varphi_1|O|\varphi_2}$ according to the type of observable. If $O=\ket{\phi}\bra{\phi}$ is a rank-one stabilizer projector, then the cross-term factorizes as $\braket{\varphi_1|O|\varphi_2} = \braket{\varphi_1|\phi}\braket{\phi|\varphi_2}$, which can be computed using stabilizer inner products. If $O$ is a general Hermitian observable $O=\sum_{i=1}^{k} c_iP_i$ decomposed into $k$ distinct Pauli operators $P_i$ with nonzero real $c_i$, then we apply the same sampled stabilizer pair to all Pauli terms:
\begin{equation}
    \braket{\varphi_1|O|\varphi_2}
    = \sum_{i=1}^{k} c_i\braket{\varphi_1|P_i|\varphi_2}
\end{equation}
As $P_i\ket{\varphi_2}$ is a stabilizer state, each term is again computed using stabilizer inner products. In Appendix~\ref{sec:algorithm-details}, we present a detailed procedure of \textsc{SampleValue} using CH-form \cite{bravyi:2019} in $\mathcal{O}(n^2)$ memory and $\mathcal{O}(n^4w + n^3k)$ time, including our inner product algorithm between the stabilizer states in $\mathcal{O}(n^3)$ time. Combining these costs with the Hoeffding sample count gives the following theorem and corollary.

\begin{theorem}\label{thm:ESSENCE}
Assume that an $n$-qubit circuit $C=U^{(w)}\cdots U^{(2)} U^{(1)}$ is given where the optimal Clifford decomposition achieving the stabilizer extent is known for each $U^{(r)}$. Then, \textnormal{\textsc{ExpVal}} solves EVE$(\epsilon, \delta)$ for Pauli or rank-one stabilizer projector $O$ with space complexity $\mathcal{O}(n^2)$ and time complexity $\mathcal{O}(\xi_\star^2\epsilon^{-2} \log(\delta^{-1})(n^4w))$, for the gate-wise stabilizer extent $\xi_\star$ of $C$.
\end{theorem} 
\begin{proof}
    Since each sample requires $\mathcal{O}(n^4w)$ time, evaluating $h=2\xi_\star^2 \epsilon^{-2}\log(2\delta^{-1})$ samples requires $\mathcal{O}(hn^4w)$ time and $\mathcal{O}(n^2)$ space.
\end{proof}

\begin{corollary}
    Assume that an $n$-qubit circuit $C=U^{(w)}\cdots U^{(2)} U^{(1)}$ is given where the optimal Clifford decomposition achieving the stabilizer extent is known for each $U^{(r)}$. 
    Assume $O=\sum_{i=1}^k{c_iP_i}$ is Pauli-sparse Hermitian, where $P_i$ is Pauli and $k = \mathcal{O}(\poly(n))$.
    Then, \textnormal{\textsc{ExpVal}} solves EVE$(\epsilon, \delta)$ for $O$ with space complexity $\mathcal{O}(\poly(n))$ and time complexity $\mathcal{O}(\xi_\star^2\epsilon^{-2} \log(\delta^{-1})w\poly(n))$.
\end{corollary}
\begin{proof}
    For each sampled stabilizer pair, all $k=\mathcal{O}(\poly(n))$ Pauli elements can be evaluated in polynomial time. Hence the total time and space complexities remain $\mathcal{O}(\xi_\star^2\epsilon^{-2}    \log(\delta^{-1})w\poly(n))$ and $\mathcal{O}(\poly(n))$, respectively.
\end{proof}

\section{CARVE Protocol}
\subsection{Variance Estimation}
A common drawback of Hoeffding-based sample-size prescriptions is that the resulting bounds are often overly conservative. The worst-case bounds fail to capture the low-variance reality of many practical circuits, which forces the simulation into an artificially high number of Monte Carlo trajectories. We therefore estimate the empirical variance and use it to obtain a tighter Bernstein-type stopping rule, by using Welford's online algorithm \cite{welford:1962} for memory efficiency. We call this procedure the CARVE protocol. CARVE operates in two stages: it first collects a small batch of $\kappa$ pilot samples to estimate the empirical standard deviation $s$. Then, it collects additional samples based on $s$ and combines them with the initial measurements to compute the final estimate.

Suppose that $L_O$ is a lower bound on the operator norm satisfying $0 < L_O \le \|O\|$.
In practice, for a Pauli-sparse observable whose canonical expansion is $O=\sum_{i=1}^kc_iP_i$, the orthogonality of distinct Pauli operators ($\Tr(P_iP_j)=2^n\delta_{ij}$) gives $2^{-n}\Tr(O^2)=\sum_i c_i^2$. Since $2^{-n}\Tr(O^2) \leq \|O\|^2$, we use the efficiently computable lower bound
\begin{equation}\label{eq:pauli-lower-bound}
    L_O = \|\mathbf{c}\|_2 =
    \sqrt{\sum_{i}c_i^2}.
\end{equation}

For the first stage, CARVE collects $\kappa$ pilot samples to compute the empirical variance $s^2$ of $ \textsc{SampleValue}$ in Algorithm~\ref{alg:sample-value}.  
We establish a high-probability upper bound for the operator-normalized true standard deviation:
\begin{equation}\label{eq:sigma-max}
    \sigma_{M} = \min\left(1, \frac{s}{L_O} + \sqrt{\frac{8\log(2\delta^{-1})}{\kappa-1}}\right).
\end{equation}

For the second stage, the protocol defines the total required sample count $h_\mathrm{CARVE}$:
\begin{equation}\label{eq:h-new}
    h_\mathrm{CARVE} = \max(\kappa, h_\kappa),
    \qquad h_\kappa = \left\lceil 2(\xi_\star^2\sigma_{M}^2+(\xi_\star +1)\epsilon/3)\epsilon^{-2}\log(4\delta^{-1}) \right\rceil
\end{equation}
The definitions of $\sigma_M$ and $h_\mathrm{CARVE}$ guarantee to solve $\textsc{EVE}(\epsilon, \delta)$ for any choice of $\kappa$. The detailed proof is provided in Appendix~\ref{sec:CARVE-proof-appendix}.

It remains to choose the pilot-batch size $\kappa$. Let 
\begin{equation}
    h_{\mathrm{Hoeffding}} = \left\lceil2\xi_\star^2\epsilon^{-2}\log(2\delta^{-1})   \right\rceil
\end{equation}
denote the Hoeffding sample count. We then choose $\kappa$ to reduce the total sample count in the low-variance regime $s/L_O \ll 1$. Then, ignoring contributions from $s/L_O$ and linear contributions from $\xi_\star$ in Eqs.~\eqref{eq:sigma-max} and \eqref{eq:h-new} give
\begin{equation}\label{eq:bernstein-result}
    \sigma_{M} \approx \sqrt{\frac{8\log(2\delta^{-1})}{\kappa}}, \qquad h_\kappa \approx \frac{16\xi_\star^2\epsilon^{-2} \log(2\delta^{-1}) \log(4\delta^{-1})}{\kappa} = \frac{8\log(4\delta^{-1})}{\kappa}h_{\mathrm{Hoeffding}}.
\end{equation} 
Balancing $\kappa \approx h_\kappa$ makes us pre-assign $\kappa$ as
\begin{equation}\label{eq:pilot}
    \kappa = \left\lceil\sqrt{8 \log(4\delta^{-1})h_{\mathrm{Hoeffding}}}\right\rceil = \left\lceil4\xi_\star\epsilon^{-1}\sqrt{\log(4\delta^{-1})\log(2\delta^{-1})}\right\rceil,
\end{equation}
minimizing the final iteration count to $h_{\mathrm{CARVE}} = \mathcal{O}(\xi_\star\epsilon^{-1}\log(\delta^{-1}))$. 

Algorithm~\ref{alg:final} outlines the complete procedure combining ESSENCE with CARVE to solve $\textsc{EVE}(\epsilon, \delta)$. Even when $s/L_O$ exceeds $1$, setting $\sigma_M = 1$ recovers the leading-order sample complexity by the Hoeffding's inequality. For clarity, we formally define the relative variance governing CARVE.
\begin{definition}[Relative Variance]
    The relative variance $\eta^2$ of a circuit $C$ and an observable $O$ is defined as 
    \begin{equation}
        \eta^2 = \frac{\sigma^2}{\|O\|^2},
    \end{equation}
    where $\sigma$ is the true standard deviation of the output from $\textnormal{\textsc{SampleValue}}(C, O)$. If $\sigma$ and the operator norm $\|O\|$ are unknown, its conservative empirical estimate $\hat{\eta}^2$ is given by
    \begin{equation}
        \hat{\eta}^2 = \frac{s^2}{L_O^2 },
    \end{equation}
    where $s$ is the empirical standard deviation collected over pilot samples, and $L_O$ is a lower bound on $\|O\|$ satisfying $0 < L_O \le \|O\|$.
\end{definition}

This formulation demonstrates that $h_{\mathrm{CARVE}}$ interpolates monotonically with respect to $\hat{\eta}$ between $\mathcal{O}(\xi_\star\epsilon^{-1}\log(\delta^{-1}))$ and $\mathcal{O}(\xi_\star^2\epsilon^{-2}\log(\delta^{-1}))$ as Eqs.~\eqref{eq:h-new}, \eqref{eq:bernstein-result}, and \eqref{eq:pilot} give
\begin{equation}
    h_{\mathrm{CARVE}} = \mathcal{O}\left(\left[\xi_\star\epsilon^{-1} + \min(\hat{\eta}^2, 1)\xi_\star^2\epsilon^{-2}\right]\log(\delta^{-1})\right).
\end{equation}

Furthermore, the implications of this variance suppression is not only restricted to ESSENCE. CARVE extends to bounded, unbiased, i.i.d. quasiprobability estimators with a known sample-range bound. As a result, it serves as a generalized tool to overcome the sampling bottlenecks that limit current classical expectation value simulation frameworks.

\begin{algorithm}[H]
    \caption{Solving \textsc{EVE}($\epsilon, \delta$) using ESSENCE with CARVE protocol}
    \label{alg:final}
    \begin{algorithmic}
        \Procedure{Welford}{$C, O$}
        \State Initialize $x \gets 0$, $t \gets 0$, $\mu \gets 0$, $\Delta \gets 0$, $M_2 \gets 0$
        \For {$j=1$ to $\kappa$}
            \State $x \gets \textsc{SampleValue}(C,O)$
            \State $t \gets t + 1$
            \State $\Delta \gets x - \mu$
            \State $\mu \gets \mu + \Delta / t$
            \State $M_2 \gets M_2 + \Delta \cdot (x - \mu)$
        \EndFor
        \State \Return $(\mu, \sqrt{M_2 / (\kappa - 1)})$
        \EndProcedure
        \Procedure{CarveExpVal}{$C, O, L_O$}
        \State $\xi_\star \gets \textsc{GetXi}(C)$
        \State $h \gets 2\xi_\star^2\epsilon^{-2}\log(2\delta^{-1})$
        \State $\kappa \gets \lceil \sqrt{8\log(4\delta^{-1})h}\rceil$
        \State $(\mu, s) \gets \textsc{Welford}(C, O)$ with $\kappa$ iterations
        \State $\hat{\eta} \gets s/L_O$
        \State $\sigma_M \gets \min(1, (\hat{\eta} + \sqrt{8\log(2\delta^{-1})/(\kappa-1)}))$
        \State $h \gets \max(\kappa, \lceil(\xi_\star^2 \sigma_M^2+(1+\xi_\star)\epsilon/3)2\epsilon^{-2}\log(4\delta^{-1})\rceil)$
        \State $a \gets \kappa \cdot \xi_\star \cdot \mu$
        \For {$j=\kappa + 1$ to $h$}
        \State $a \gets a + \xi_\star \cdot \textsc{SampleValue}(C,O)$
        \EndFor
        \State \Return $a/h$
        \EndProcedure
    \end{algorithmic}
\end{algorithm}

\subsection{Performance Analysis}
We have shown that the time complexity scales down from $\mathcal{O}(\xi_\star^2)$ to $\mathcal{O}(\xi_\star)$ in the best-case scenario. To demonstrate that this favorable scaling is not a rare anomaly, we establish rigorous theoretical upper bounds on this expected variance suppression in two key groups of circuits: idealized global unitaries and practical fault-tolerant architectures.

  We first consider the idealized case of global Haar-random unitaries. The detailed proof is provided in Appendix~\ref{sec:haar-CARVE}.

\begin{restatable}[Variance Suppression in Single Unitary]{theorem}{varianceThm}
\label{thm:variance-estimation}
    Let $U$ be a Haar random $n$-qubit unitary, and let $O$ be any non-identity Pauli. Assume that the optimal Clifford decomposition for $U$ is always given. Then, the expected relative variance $\eta^2 = \sigma^2/\|O\|^2$ of the single gate $U$ and the observable $O$ satisfies 
    \begin{equation}
        \mathbb{E}_{U \sim \text{Haar}} [\eta^2] \le \frac{2^n}{4^n - 1}.
    \end{equation}
\end{restatable}

While Theorem~\ref{thm:variance-estimation} establishes the theoretical limits of variance suppression under idealized global randomness, practical fault-tolerant quantum computation is often modeled using highly structured circuits composed predominantly of Clifford operations interleaved with a discrete number of $T$-gates. Specifically, we consider sparse-$T$ circuits which contains exactly one $T$-gate per Clifford layer. We demonstrate that CARVE exponentially suppresses the relative variance with respect to both the qubit count $n$ and the $T$-gate count $t$ under this specific model as follows:
\begin{restatable}[Variance Suppression in Clifford + Sparse T circuits]{theorem}{varianceNISQthm}
\label{thm:nisq-thm}
    For a random $n$-qubit circuit $C$ defined as
    \begin{equation}
        C = U^{(t)}T_{k_t}U^{(t-1)}T_{k_{t-1}}\cdots U^{(1)}T_{k_1}U^{(0)},
    \end{equation}
    where every Clifford gate $U^{(r)}$ and qubit index $k_r \in [n]$ is sampled independently and uniformly, the expected relative variance $\eta^2 = \sigma^2 / \|O\|^2$ of the circuit $C$ and any observable $O$ satisfies
    \begin{equation}
        \mathbb{E}_{C}[\eta^2] < 0.5^n + 0.75^t.
    \end{equation}
\end{restatable}
The proof of Theorem~\ref{thm:nisq-thm} is given in Appendix~\ref{sec:clifford-t-CARVE}. Since the relative variance rapidly converges to zero, it suggests that ESSENCE with CARVE achieves a substantial algorithmic advantage for a wide range of practical quantum circuits. For instance, even at a modest scale of $n=30$ qubits and $t=10$ T-gates, the expected relative variance is strictly bounded above by $0.057$. It remains an open problem for circuits with dense $T$-gate configuration. In the following section, we provide numerical evidence that variance suppression translates into practice even for moderate-rank circuits. 
    
\section{Simulation Results}\label{sec:results}
For all simulated results, we fixed the error confidence parameters to $\epsilon = 0.2$ and $\delta = 0.2$, ensuring that the total error is bounded by $\epsilon\|O\|$ with a probability of at least $1-\delta$. Simulations were executed using an Intel Core Ultra 9 285K CPU, by distributing independent executions of \textsc{SampleValue} in parallel across 24 cores. The source code used to generate these results is provided in the Code and Data Availability section. 
More detailed simulation results can be found in Appendix~\ref{sec:results-appendix}.

\begin{figure}[H]
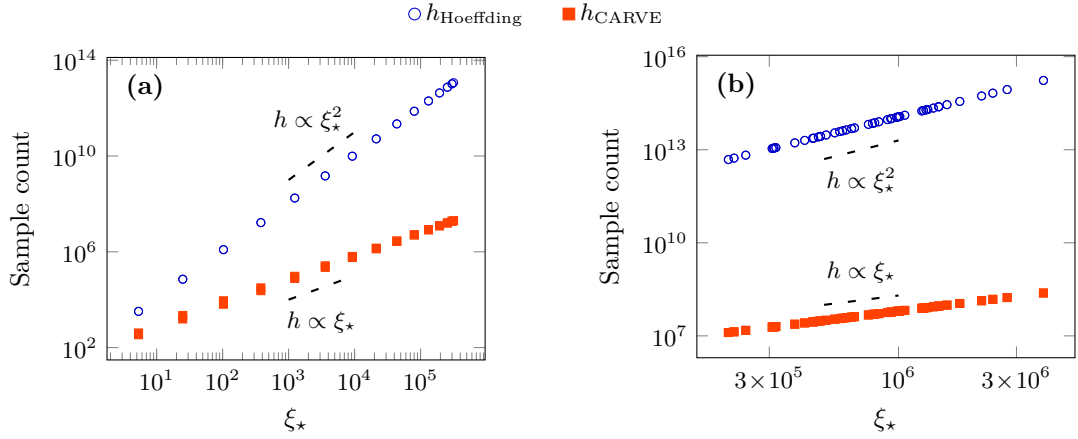

    \centering
    \ref*{sharedlegend}
    \vspace{0.0cm}
    
    \begin{subfigure}{.48\textwidth}
        \centering
        \phantomsubcaption
        \label{fig:qaoa}
        \graphscatterA{data/details_N60D4.csv}{5}{4}
    \end{subfigure}%
    \hfill
    \begin{subfigure}{.48\textwidth}
        \centering
        \phantomsubcaption
        \label{fig:qke}
        \graphscatterB{data/details_QKE_n64.csv}{5}{4}
    \end{subfigure}
    \caption{Sample counts $h_{\mathrm{Hoeffding}}$ and $h_\mathrm{CARVE}$ as functions of $\xi_\star$ on logarithmic axes. Dashed segments indicate the reference scalings $h\propto\xi_\star$ and $h\propto\xi_\star^2$. (a) QAOA benchmark for the N60D4 MaxE3LIN2 instance with fixed $\beta=\pi/4$, $\epsilon=0.2$, and $\delta=0.2$. (b) Quantum-kernel benchmark for 100 kernel values $K(x',x)$ constructed from 10 independent random vectors  sampled uniformly from $[0,2\pi]^{64}$.}
    \label{fig:result}
\end{figure}

\subsection{Evaluation on the QAOA Ansatz}\label{sec:results-qaoa}

We evaluated the QAOA ansatz \cite{farhi2014quantum} applied to the MaxE3LIN2 problem. It provides an ideal benchmarking environment because the resulting quantum circuits are dominated by Clifford gates with a moderate count of non-Clifford $R_z$ rotations and the observable is Pauli-sparse, making it highly suitable for ESSENCE. The objective of the problem is to find a Boolean assignment $z \in \set{-1, 1}^n$ that maximizes the cost function:
\begin{equation}
    \mathrm{cost}(z_1, z_2, \dots, z_n) = \sum_{1 \le u < v < w \le n} d_{uvw} z_u z_v z_w
\end{equation}
where $d_{uvw} \in \{0, \pm 1\}$. The MaxE3LIN2 problem is said to have a degree $D$ if every variable $z_u$ appears in exactly $D$ nonzero terms. We denote a problem instance with $n$ qubits and degree $D$ as N$n$D$D$.

The QAOA approach involves defining a cost Hamiltonian $O_c = \sum d_{uvw} Z_u Z_v Z_w$ and generating a parameterized quantum state via the circuit $U(\beta, \gamma)$:
\begin{equation}
    U(\beta, \gamma) = e^{-i\beta (X_1 + X_2 + \cdots+X_n)} e^{-i(\gamma/2) O_c} H^{\otimes n}
\end{equation}
The corresponding QAOA objective value is the expectation value of the cost Hamiltonian in the variational state:
\begin{equation}
    E(\beta, \gamma) = \braket{0^{\otimes n} | U(\beta, \gamma)^\dagger O_c U(\beta, \gamma) | 0^{\otimes n}}
\end{equation}

Following established conventions for this benchmark \cite{bravyi:2019}, we fixed $\beta = \pi/4$ and analyzed the function $E(\gamma) := E(\pi/4, \gamma)$. At this fixed $\beta$, the circuit consists entirely of Clifford operations and exactly $nD/3$ non-Clifford $Z$-rotations given by $e^{i(\gamma/2)Z}$. Due to the symmetry of the expectation value, where $E(\gamma) = -E(-\gamma)$ and $E(\gamma) = \pm E(\pi - \gamma)$, evaluating the range $\gamma \in [0, \pi/2]$ is sufficient to fully characterize the circuit. For our simulation, we evaluated $E(\gamma)$ at 31 equally spaced points, $\gamma = \pi l/60$ for $l = 0, 1, \dots, 30$. In CARVE, we used Eq.~\eqref{eq:pauli-lower-bound} for the operator norm lower bound.

\begin{figure}[htbp]
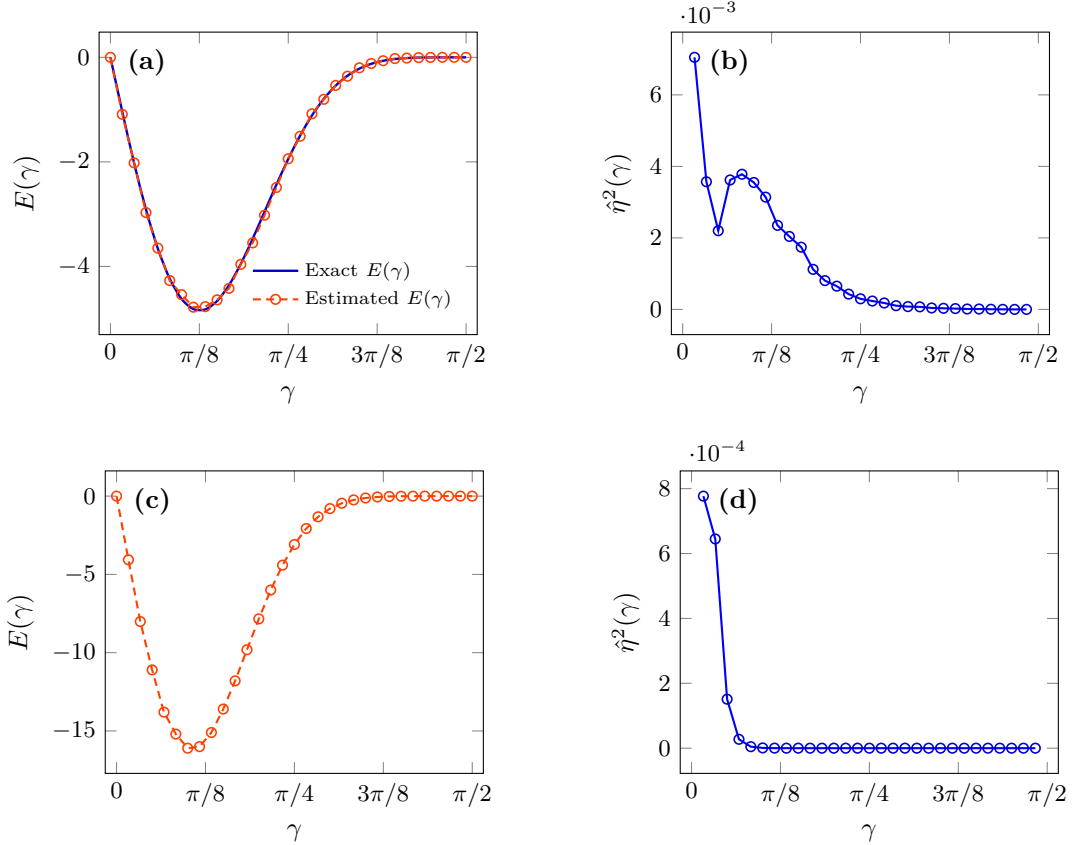

    \centering
    \begin{subfigure}{.48\textwidth}
        \centering
        \phantomsubcaption
        \label{fig:qaoa-show-n20-obj}
        \qaoaA{data/details_N20D3.csv}{5}{4}
    \end{subfigure}
    \hfill
    \begin{subfigure}{.48\textwidth}
        \centering
        \phantomsubcaption
        \label{fig:qaoa-show-n20-var}
        \qaoaB{data/details_N20D3.csv}{5}{4}
    \end{subfigure}
    \vspace{1em}

    \centering
    \begin{subfigure}{.48\textwidth}
        \centering
        \phantomsubcaption
        \label{fig:qaoa-show-n60-obj}
        \qaoaC{data/details_N60D4.csv}{5}{4}
    \end{subfigure}
    \hfill
    \begin{subfigure}{.48\textwidth}
        \centering
        \phantomsubcaption
        \label{fig:qaoa-show-n60-var}
        \qaoaD{data/details_N60D4.csv}{5}{4}
    \end{subfigure}        
    \caption{
    QAOA results for MaxE3LIN2 instances N20D3 and N60D4.
    (a) Estimated and exact objective value $E(\gamma)$ for the N20D3 instance. (b) Conservative empirical relative variance $\hat{\eta}^2$ for the N20D3 instance. (c) Estimated objective value $E(\gamma)$ for the N60D4 instance. (d) Conservative empirical relative variance $\hat{\eta}^2$ for the N60D4 instance.}
    \label{fig:qaoa-show}
\end{figure}

For the result of the N60D4 problem depicted in Fig.~\ref{fig:qaoa}, CARVE achieves a quadratic reduction in the number of sampled stabilizer pairs required to estimate the expectation value when compared to the Hoeffding's. Moreover, the computational runtime exhibits a linear scaling with the gate-wise stabilizer extent, $\xi_\star$. Simulating the circuit across the parameter sweep took approximately 9 hours, peaking at 4,500 seconds per configuration. If the simulation had instead relied on Hoeffding bounds, extrapolations based on the required iterations indicate the N60D4 simulation would have taken approximately 424 years. This demonstrates that the bounds were overly conservative in practice. By tightening these bounds, we obtained a significant reduction in practical computational overhead. Furthermore, for the N20D3 instance depicted in Fig.~\ref{fig:qaoa-show-n20-obj}, ESSENCE with CARVE estimates matched the state-vector ground truth within the prescribed tolerance $\epsilon\|O\|$. Figures~\ref{fig:qaoa-show-n20-var} and \ref{fig:qaoa-show-n60-var} show that conservative empirical relative variance $\hat{\eta}^2 = s^2/L_O^2$ is suppressed well. Normalizing by the exact norm $\|O\|$ yields an even smaller value. This empirically supports Theorem~\ref{thm:nisq-thm}.

\subsection{Evaluation of Quantum Kernels}\label{sec:results-qsvm}

To further demonstrate the versatility of our framework, we benchmark its performance on quantum kernel estimation, a foundational subroutine for QSVM \cite{QSVM}. In QSVM, a quantum kernel computes a similarity measure between classical data points by mapping them into a high-dimensional quantum Hilbert space. 
Mathematically, the quantum kernel element between two data points, $x$ and $x'$, is defined by the overlap of their corresponding quantum feature states:
\begin{equation}
    K(x, x') = \abs{\braket{\phi(x)|\phi(x')}}^2
\end{equation}
These feature states are prepared by applying a parameterized unitary data encoding circuit $U(x)$ to an initial state, such that $\ket{\phi(x)} = U(x)\ket{+^{\otimes n}}$. 
Equivalently, by defining the composite circuit
\begin{equation}
    V(x, x') = H^{\otimes n}U^\dagger(x)U(x')H^{\otimes n},
\end{equation}
the same kernel element can be written as the expectation value of the projector $\Pi_0=\ket{0^{\otimes n}}\bra{0^{\otimes n}}$, yielding
\begin{equation}
    K(x, x') = \braket{0^{\otimes n}|V(x, x')^\dagger \Pi_0 V(x, x') | 0^{\otimes n}}.
\end{equation}

This observable-driven formulation aligns naturally with the operational design of ESSENCE. At the same time, classical simulation of quantum-kernel circuits can become costly when the feature map contains many data-dependent non-Clifford rotations.
To benchmark this scenario, we employed a Pauli feature map utilizing $Z$ and $ZZ$ rotations for the encoder $U$. We generated 10 independent classical data vectors, $x \in [0, 2\pi]^n$, with each component $x[j]$ drawn from a uniform distribution $\mathcal{U}(0, 2\pi)$. We then encoded each sample into an $n$-qubit feature state via the unitary:
\begin{equation}
    U(x) = \exp\left( i \sum_{j=0}^{n-1} x[j] Z_j + i \sum_{j=0}^{n-2} (\pi - x[j])(\pi - x[j+1]) Z_j Z_{j+1} \right).
\end{equation}
From these 10 generated samples $x_0, x_1, \cdots, x_9$, we computed the kernel elements $K(x_i, x_j)$ for all 100 ordered pairs.

\begin{figure}[H]
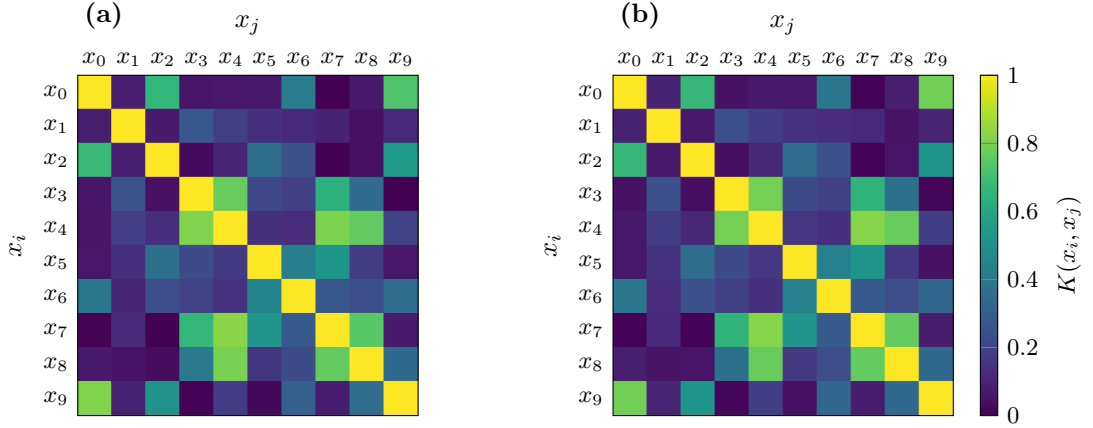

    \centering
    \begin{subfigure}[t]{.48\textwidth}
        \centering 
        \phantomsubcaption
        \label{fig:qke-show-n2-res}
        \qkeA{data/details_QKE_n2.csv}{4.5}{4.5}
    \end{subfigure}
    \hfill
    \begin{subfigure}[t]{.48\textwidth}
        \centering
        \phantomsubcaption
        \label{fig:qke-show-n2-truth}
        \qkeB{data/details_QKE_n2.csv}{4.5}{4.5}
    \end{subfigure}
    \caption{
        Quantum kernel matrices $K(x_i, x_j)$ for 10 independent vectors $x_i$ drawn uniformly from $[0,2\pi]^2$ in the two-qubit setting ($n=2$).
        (a) Estimated kernel matrix.
        (b) Exact kernel matrix.
        Rows and columns use the same ordering of data points, and both panels use the same linear color scale from 0 to 1.}
    \label{fig:qke-show}
\end{figure}

\begin{figure}[H]
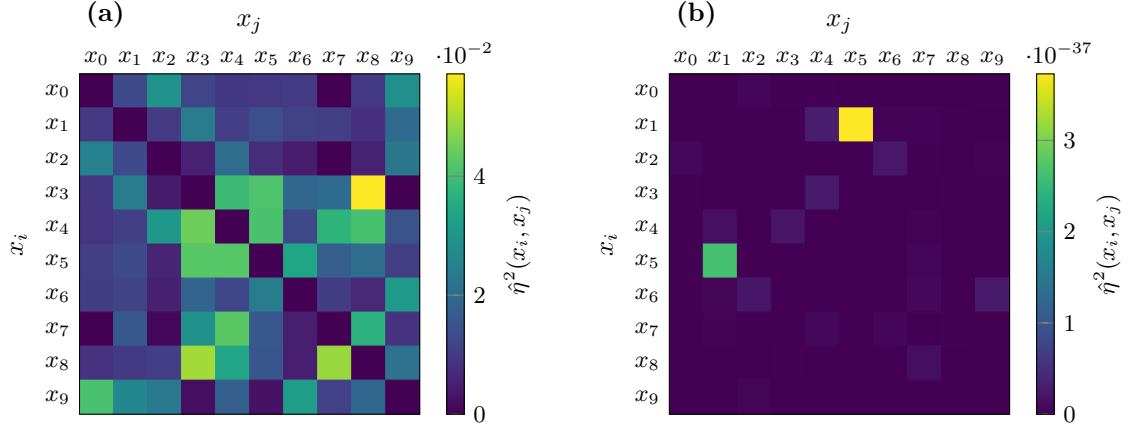

    \centering
    \begin{subfigure}[t]{.48\textwidth}
        \centering
        \phantomsubcaption
        \label{fig:qke-show-n2-stddev}
        \qkestddev{data/details_QKE_n2.csv}{4.5}{4.5}{(a)}
    \end{subfigure}
    \hfill
    \begin{subfigure}[t]{.48\textwidth}
        \centering
        \phantomsubcaption
        \label{fig:qke-show-n64-stddev}
        \qkestddev{data/details_QKE_n64.csv}{4.5}{4.5}{(b)}
    \end{subfigure}
    \caption{
    Empirical relative variance matrices $\hat{\eta}^2(x_i, x_j)$ for (a) two qubits ($n=2$) and (b) 64 qubits ($n=64$).
    Rows and columns correspond to the same ordering of data points. The panels use independent linear color scales, as indicated by their respective colorbars.}
    \label{fig:qke-stddev}
\end{figure}

As illustrated in Fig.~\ref{fig:qke}, the quadratic sample reduction and linear runtime scaling with $\xi_\star$ established in the QAOA benchmark persist here. For the $n=2$ case depicted in Fig.~\ref{fig:qke-show}, a comparison between the estimated heatmaps and the exact state-vector ground truth confirms the accuracy of our approach, with entrywise absolute errors remaining below $\epsilon = 0.2$. Furthermore, Fig.~\ref{fig:qke-stddev} demonstrates that the empirical relative variance $\hat{\eta}^2 =s^2/\|O\|^2$ is consistently small for both $n=2$ and $n=64$, and indicates a decreasing trend as the system size $n$ increases.

\section{Conclusion}
In this work, we introduced ESSENCE, an efficient simulation framework that mitigates the severe computational bottlenecks associated with estimating expectation values in quantum circuits containing non-Clifford gates. While stochastic methods like DCS offer theoretical pathways, their practical utility has historically been restricted by exponential decomposition scaling and conservative sampling bounds. ESSENCE addresses these limitations by integrating SoC to expectation value estimation.

We further incorporated the CARVE protocol, which allows the simulation to dynamically monitor empirical sample variance rather than relying on overly pessimistic Hoeffding's inequality. By doing so, the simulation terminates once its empirical confidence bound certifies the prescribed $(\epsilon,\delta)$ accuracy, effectively reducing unnecessary Monte Carlo iterations. These advancements maintain a practical space complexity of $\mathcal{O}(n^2)$ while significantly reducing the computational time complexity in practical cases. Our simulations on the QAOA ansatz for MaxE3LIN2 and on quantum-kernel circuits used in QSVMs support these theoretical and algorithmic advantages.

Moreover, CARVE can be broadly generalized. The empirical variance estimation subroutine we developed is not limited to our specific decomposition method, but is adaptable to similar sampling-based simulations involving stabilizer states. Consequently, this approach provides a generalized tool to overcome computational bottlenecks in broader quasi-probabilistic propagation algorithms well beyond the immediate scope of ESSENCE.

\section*{Code and Data Availability}
The source code, datasets, and scripts required to reproduce the QAOA and quantum kernel estimation benchmarks presented in this work are openly available on GitHub at \url{https://github.com/snu-quiqcl/essence_carve}.

\section*{Acknowledgements}
This work was supported by the National Research Foundation of Korea~(NRF) grants (No. RS-2024-00442855, No. RS-2024-00413957) and the Institute for Information \& Communications Technology Planning \& Evaluation~(IITP) grant (No. RS-2022-II221040), all funded by the Korean government~(MSIT). H.K. is supported by KIAS individual grant No. CG085302 at the Korea Institute for Advanced Study.

\bibliography{refs}

\begin{thebibliography}{10}

\bibitem{MPS:2003}
Guifr{\'e} Vidal.
\newblock ``Efficient classical simulation of slightly entangled quantum computations''.
\newblock \href{https://dx.doi.org/10.1103/PhysRevLett.91.147902}{Phys. Rev. Lett. {\bf 91}, 147902}~(2003).

\bibitem{bravyi:2016}
Sergey Bravyi and David Gosset.
\newblock ``Improved classical simulation of quantum circuits dominated by {Clifford} gates''.
\newblock \href{https://dx.doi.org/10.1103/PhysRevLett.116.250501}{Phys. Rev. Lett. {\bf 116}, 250501}~(2016).

\bibitem{XU20254104}
Xiaosi Xu, Simon Benjamin, Jianxin Chen, Jinzhao Sun, Xiao Yuan, and Pan Zhang.
\newblock ``A {Herculean} task: classical simulation of quantum computers''.
\newblock \href{https://dx.doi.org/10.1016/j.scib.2025.10.016}{Sci. Bull. {\bf 70}, 4104--4112}~(2025).

\bibitem{cerezo2021variational}
M.~Cerezo, Andrew Arrasmith, Ryan Babbush, Simon~C. Benjamin, Suguru Endo, Keisuke Fujii, Jarrod~R. McClean, Kosuke Mitarai, Xiao Yuan, Lukasz Cincio, and Patrick~J. Coles.
\newblock ``Variational quantum algorithms''.
\newblock \href{https://dx.doi.org/10.1038/s42254-021-00348-9}{Nat. Rev. Phys. {\bf 3}, 625--644}~(2021).

\bibitem{peruzzo2014variational}
Alberto Peruzzo, Jarrod McClean, Peter Shadbolt, Man-Hong Yung, Xiao-Qi Zhou, Peter~J. Love, Al{\'a}n Aspuru-Guzik, and Jeremy~L. O'Brien.
\newblock ``A variational eigenvalue solver on a photonic quantum processor''.
\newblock \href{https://dx.doi.org/10.1038/ncomms5213}{Nat. Commun. {\bf 5}, 4213}~(2014).

\bibitem{mcclean2016theory}
Jarrod~R. McClean, Jonathan Romero, Ryan Babbush, and Al{\'a}n Aspuru-Guzik.
\newblock ``The theory of variational hybrid quantum-classical algorithms''.
\newblock \href{https://dx.doi.org/10.1088/1367-2630/18/2/023023}{New J. Phys. {\bf 18}, 023023}~(2016).

\bibitem{zhao2020measurement}
Andrew Zhao, Andrew Tranter, William~M. Kirby, Shu~Fay Ung, Akimasa Miyake, and Peter~J. Love.
\newblock ``Measurement reduction in variational quantum algorithms''.
\newblock \href{https://dx.doi.org/10.1103/PhysRevA.101.062322}{Phys. Rev. A {\bf 101}, 062322}~(2020).

\bibitem{huggins2021efficient}
William~J. Huggins, Jarrod~R. McClean, Nicholas~C. Rubin, Zhang Jiang, Nathan Wiebe, K.~Birgitta Whaley, and Ryan Babbush.
\newblock ``Efficient and noise resilient measurements for quantum chemistry on near-term quantum computers''.
\newblock \href{https://dx.doi.org/10.1038/s41534-020-00341-7}{npj Quantum Inf. {\bf 7}, 23}~(2021).

\bibitem{temme2017error}
Kristan Temme, Sergey Bravyi, and Jay~M. Gambetta.
\newblock ``Error mitigation for short-depth quantum circuits''.
\newblock \href{https://dx.doi.org/10.1103/PhysRevLett.119.180509}{Phys. Rev. Lett. {\bf 119}, 180509}~(2017).

\bibitem{kandala2017hardware}
Abhinav Kandala, Antonio Mezzacapo, Kristan Temme, Maika Takita, Markus Brink, Jerry~M. Chow, and Jay~M. Gambetta.
\newblock ``Hardware-efficient variational quantum eigensolver for small molecules and quantum magnets''.
\newblock \href{https://dx.doi.org/10.1038/nature23879}{Nature {\bf 549}, 242--246}~(2017).

\bibitem{farhi2014quantum}
Edward Farhi, Jeffrey Goldstone, and Sam Gutmann.
\newblock ``A quantum approximate optimization algorithm''~(2014).
\newblock  \href{http://arxiv.org/abs/1411.4028}{arXiv:1411.4028}.

\bibitem{biamonte2017quantum}
Jacob Biamonte, Peter Wittek, Nicola Pancotti, Patrick Rebentrost, Nathan Wiebe, and Seth Lloyd.
\newblock ``Quantum machine learning''.
\newblock \href{https://dx.doi.org/10.1038/nature23474}{Nature {\bf 549}, 195--202}~(2017).

\bibitem{begusic:2024}
Tomislav Begu\v{s}i\'{c}, Johnnie Gray, and Garnet Kin-Lic Chan.
\newblock ``Fast and converged classical simulations of evidence for the utility of quantum computing before fault tolerance''.
\newblock \href{https://dx.doi.org/10.1126/sciadv.adk4321}{Sci. Adv. {\bf 10}, eadk4321}~(2024).

\bibitem{kim:2023}
Youngseok Kim, Andrew Eddins, Sajant Anand, Ken~Xuan Wei, Ewout van~den Berg, Sami Rosenblatt, Hasan Nayfeh, Yantao Wu, Michael Zaletel, Kristan Temme, and Abhinav Kandala.
\newblock ``Evidence for the utility of quantum computing before fault tolerance''.
\newblock \href{https://dx.doi.org/10.1038/s41586-023-06096-3}{Nature {\bf 618}, 500--505}~(2023).

\bibitem{QSVM}
Vojt{\v{e}}ch Havl{\'\i}{\v{c}}ek, Antonio~D. C{\'o}rcoles, Kristan Temme, Aram~W. Harrow, Abhinav Kandala, Jerry~M. Chow, and Jay~M. Gambetta.
\newblock ``Supervised learning with quantum-enhanced feature spaces''.
\newblock \href{https://dx.doi.org/10.1038/s41586-019-0980-2}{Nature {\bf 567}, 209--212}~(2019).

\bibitem{NISQalgo}
Kishor Bharti, Alba Cervera-Lierta, Thi~Ha Kyaw, Tobias Haug, Sumner Alperin-Lea, Abhinav Anand, Matthias Degroote, Hermanni Heimonen, Jakob~S. Kottmann, Tim Menke, Wai-Keong Mok, Sukin Sim, Leong-Chuan Kwek, and Al{\'a}n Aspuru-Guzik.
\newblock ``Noisy intermediate-scale quantum algorithms''.
\newblock \href{https://dx.doi.org/10.1103/RevModPhys.94.015004}{Rev. Mod. Phys. {\bf 94}, 015004}~(2022).

\bibitem{gottesman1997stabilizer}
Daniel Gottesman.
\newblock ``Stabilizer codes and quantum error correction''.
\newblock \href{https://dx.doi.org/10.48550/arXiv.quant-ph/9705052}{PhD thesis}.
\newblock California Institute of Technology.
\newblock ~(1997).
\newblock  \href{http://arxiv.org/abs/quant-ph/9705052}{arXiv:quant-ph/9705052}.

\bibitem{gottesman:1998}
Daniel Gottesman.
\newblock ``The {Heisenberg} representation of quantum computers''~(1998).
\newblock  \href{http://arxiv.org/abs/quant-ph/9807006}{arXiv:quant-ph/9807006}.

\bibitem{valiant2002quantum}
Leslie~G. Valiant.
\newblock ``Quantum circuits that can be simulated classically in polynomial time''.
\newblock \href{https://dx.doi.org/10.1137/S0097539700377025}{SIAM J. Comput. {\bf 31}, 1229--1254}~(2002).

\bibitem{terhal2002classical}
Barbara~M. Terhal and David~P. DiVincenzo.
\newblock ``Classical simulation of noninteracting-fermion quantum circuits''.
\newblock \href{https://dx.doi.org/10.1103/PhysRevA.65.032325}{Phys. Rev. A {\bf 65}, 032325}~(2002).

\bibitem{jozsa2008matchgates}
Richard Jozsa and Akimasa Miyake.
\newblock ``Matchgates and classical simulation of quantum circuits''.
\newblock \href{https://dx.doi.org/10.1098/rspa.2008.0189}{Proc. R. Soc. A {\bf 464}, 3089--3106}~(2008).

\bibitem{shi2006classical}
Y.-Y. Shi, L.-M. Duan, and G.~Vidal.
\newblock ``Classical simulation of quantum many-body systems with a tree tensor network''.
\newblock \href{https://dx.doi.org/10.1103/PhysRevA.74.022320}{Phys. Rev. A {\bf 74}, 022320}~(2006).

\bibitem{koukoulekidis2022faster}
Nikolaos Koukoulekidis, Hyukjoon Kwon, Hyejung~H. Jee, David Jennings, and M.~S. Kim.
\newblock ``Faster born probability estimation via gate merging and frame optimisation''.
\newblock \href{https://dx.doi.org/10.22331/q-2022-10-13-838}{Quantum {\bf 6}, 838}~(2022).

\bibitem{ipek2026phase}
Selman Ipek, Atak~Talay Yucel, Farzad Shahi, Cagdas Ozdemir, and Cihan Okay.
\newblock ``Phase-space tableau simulation for quantum computation''.
\newblock \href{https://dx.doi.org/10.1103/PhysRevA.113.032409}{Phys. Rev. A {\bf 113}, 032409}~(2026).

\bibitem{wang2022possibilistic}
Daochen Wang.
\newblock ``Possibilistic simulation of quantum circuits by classical circuits''.
\newblock \href{https://dx.doi.org/10.1103/PhysRevA.106.062430}{Phys. Rev. A {\bf 106}, 062430}~(2022).

\bibitem{cichy2025classical}
Simon Cichy, Paul~K. Faehrmann, Lennart Bittel, Jens Eisert, and Hakop Pashayan.
\newblock ``Classical simulation of noisy quantum circuits via locally entanglement-optimal unravelings''~(2025).
\newblock  \href{http://arxiv.org/abs/2508.05745}{arXiv:2508.05745}.

\bibitem{peres2024non}
Filipa C.~R. Peres, Rafael Wagner, and Ernesto~F. Galv{\~a}o.
\newblock ``Non-stabilizerness and entanglement from cat-state injection''.
\newblock \href{https://dx.doi.org/10.1088/1367-2630/ad1b80}{New J. Phys. {\bf 26}, 013051}~(2024).

\bibitem{hahn2025bridging}
Oliver Hahn, Giulia Ferrini, and Ryuji Takagi.
\newblock ``Bridging magic and non-{Gaussian} resources via {Gottesman-Kitaev-Preskill} encoding''.
\newblock \href{https://dx.doi.org/10.1103/PRXQuantum.6.010330}{PRX Quantum {\bf 6}, 010330}~(2025).

\bibitem{hahn2025classical}
Oliver Hahn, Ryuji Takagi, Giulia Ferrini, and Hayata Yamasaki.
\newblock ``Classical simulation and quantum resource theory of non-{Gaussian} optics''.
\newblock \href{https://dx.doi.org/10.22331/q-2025-10-13-1881}{Quantum {\bf 9}, 1881}~(2025).

\bibitem{pashayan2022fast}
Hakop Pashayan, Oliver Reardon-Smith, Kamil Korzekwa, and Stephen~D. Bartlett.
\newblock ``Fast estimation of outcome probabilities for quantum circuits''.
\newblock \href{https://dx.doi.org/10.1103/PRXQuantum.3.020361}{PRX Quantum {\bf 3}, 020361}~(2022).

\bibitem{hamaguchi2025faster}
Hiroki Hamaguchi, Kou Hamada, Naoki Marumo, and Nobuyuki Yoshioka.
\newblock ``Faster computation of nonstabilizerness''.
\newblock \href{https://dx.doi.org/10.1103/PhysRevApplied.23.014069}{Phys. Rev. Appl. {\bf 23}, 014069}~(2025).

\bibitem{burgholzer2021random}
Lukas Burgholzer, Richard Kueng, and Robert Wille.
\newblock ``Random stimuli generation for the verification of quantum circuits''.
\newblock In Proc. 26th Asia South Pac. Des. Autom. Conf.
\newblock \href{https://dx.doi.org/10.1145/3394885.3431590}{Pages 767--772}.
\newblock ~(2021).

\bibitem{dias2026optimal}
Beatriz Dias, Jan~Lukas Bosse, and James~R. Seddon.
\newblock ``Optimal and improved gate decompositions for accelerated classical simulation of near-{Gaussian} fermionic circuits''~(2026).
\newblock  \href{http://arxiv.org/abs/2603.18869}{arXiv:2603.18869}.

\bibitem{yashin2025further}
Vsevolod~I. Yashin, Vladimir~V. Yatsulevich, Aleksey~K. Fedorov, and Evgeniy~O. Kiktenko.
\newblock ``Further improvements to stabilizer simulation theory: classical rewriting of {CSS}-preserving stabilizer circuits, quadratic form expansions of stabilizer operations, and framed hidden variable models''~(2025).
\newblock  \href{http://arxiv.org/abs/2511.05478}{arXiv:2511.05478}.

\bibitem{hakkaku2021comparative}
Shigeo Hakkaku and Keisuke Fujii.
\newblock ``Comparative study of sampling-based simulation costs of noisy quantum circuits''.
\newblock \href{https://dx.doi.org/10.1103/PhysRevApplied.15.064027}{Phys. Rev. Appl. {\bf 15}, 064027}~(2021).

\bibitem{reardon2024improved}
Oliver Reardon-Smith, Micha{\l} Oszmaniec, and Kamil Korzekwa.
\newblock ``Improved simulation of quantum circuits dominated by free fermionic operations''.
\newblock \href{https://dx.doi.org/10.22331/q-2024-12-04-1549}{Quantum {\bf 8}, 1549}~(2024).

\bibitem{hakkaku2021sampling}
Shigeo Hakkaku, Kosuke Mitarai, and Keisuke Fujii.
\newblock ``Sampling-based quasiprobability simulation for fault-tolerant quantum error correction on the surface codes under coherent noise''.
\newblock \href{https://dx.doi.org/10.1103/PhysRevResearch.3.043130}{Phys. Rev. Res. {\bf 3}, 043130}~(2021).

\bibitem{piveteau2025simulating}
Christophe Piveteau.
\newblock ``Simulating quantum circuits with restricted quantum computers''~(2025).
\newblock  \href{http://arxiv.org/abs/2503.21773}{arXiv:2503.21773}.

\bibitem{ballarin2025optimal}
Marco Ballarin, Pietro Silvi, Simone Montangero, and Daniel Jaschke.
\newblock ``Optimal sampling of tensor networks targeting wave function's fast decaying tails''.
\newblock \href{https://dx.doi.org/10.48550/arXiv.2401.10330}{Quantum {\bf 9}, 1714}~(2025).

\bibitem{yuan2024virtual}
Xiao Yuan, Bartosz Regula, Ryuji Takagi, and Mile Gu.
\newblock ``Virtual quantum resource distillation''.
\newblock \href{https://dx.doi.org/10.1103/PhysRevLett.132.050203}{Phys. Rev. Lett. {\bf 132}, 050203}~(2024).

\bibitem{huang2025nonstabilizerness}
Jiale Huang, Xiangjian Qian, and Mingpu Qin.
\newblock ``Nonstabilizerness entanglement entropy: A measure of hardness in the classical simulation of quantum many-body systems with tensor network states''.
\newblock \href{https://dx.doi.org/10.1103/PhysRevA.112.012425}{Phys. Rev. A {\bf 112}, 012425}~(2025).

\bibitem{haug2023stabilizer}
Tobias Haug and Lorenzo Piroli.
\newblock ``Stabilizer entropies and nonstabilizerness monotones''.
\newblock \href{https://dx.doi.org/10.22331/q-2023-08-28-1092}{Quantum {\bf 7}, 1092}~(2023).

\bibitem{zhang2026enhancing}
Ruiqi Zhang, Fuchuan Wei, and Zhaohui Wei.
\newblock ``Enhancing classical simulation with noisy quantum devices''~(2026).
\newblock  \href{http://arxiv.org/abs/2601.08772}{arXiv:2601.08772}.

\bibitem{paviglianiti2025estimating}
Alessio Paviglianiti, Guglielmo Lami, Mario Collura, and Alessandro Silva.
\newblock ``Estimating nonstabilizerness dynamics without simulating it''.
\newblock \href{https://dx.doi.org/10.1103/PRXQuantum.6.030320}{PRX Quantum {\bf 6}, 030320}~(2025).

\bibitem{saxena2022quantifying}
Gaurav Saxena and Gilad Gour.
\newblock ``Quantifying multiqubit magic channels with completely stabilizer-preserving operations''.
\newblock \href{https://dx.doi.org/10.1103/PhysRevA.106.042422}{Phys. Rev. A {\bf 106}, 042422}~(2022).

\bibitem{hakkaku2022quantifying}
Shigeo Hakkaku, Yuichiro Tashima, Kosuke Mitarai, Wataru Mizukami, and Keisuke Fujii.
\newblock ``Quantifying fermionic nonlinearity of quantum circuits''.
\newblock \href{https://dx.doi.org/10.1103/PhysRevResearch.4.043100}{Phys. Rev. Res. {\bf 4}, 043100}~(2022).

\bibitem{wang2023universal}
Dong-Sheng Wang.
\newblock ``Universal resources for quantum computing''.
\newblock \href{https://dx.doi.org/10.1088/1572-9494/ad07d6}{Commun. Theor. Phys. {\bf 75}, 125101}~(2023).

\bibitem{marshall2023simulation}
Jeffrey Marshall and Namit Anand.
\newblock ``Simulation of quantum optics by coherent state decomposition''.
\newblock \href{https://dx.doi.org/10.1364/OPTICAQ.504311}{Optica Quantum {\bf 1}, 78--93}~(2023).

\bibitem{miniskar2026q}
Narasinga~Rao Miniskar, Mohammad Alaul~Haque Monil, Elaine Wong, Vicente~Leyton Ortega, Jeffrey~S. Vetter, Seth~R. Johnson, and Travis Humble.
\newblock ``Q-{IRIS}: The evolution of the {IRIS} task-based runtime to enable classical-quantum workflows''.
\newblock In Proc. Supercomput. Asia Int. Conf. High Perform. Comput. Asia Pac. Reg. Workshop.
\newblock \href{https://dx.doi.org/10.1145/3784828.3785240}{Pages 323--329}.
\newblock ~(2026).

\bibitem{jing2025circuit}
Mingrui Jing, Chengkai Zhu, and Xin Wang.
\newblock ``Circuit knitting facing exponential sampling-overhead scaling bounded by entanglement cost''.
\newblock \href{https://dx.doi.org/10.1103/PhysRevA.111.012433}{Phys. Rev. A {\bf 111}, 012433}~(2025).

\bibitem{luthra2025unlocking}
Surabhi Luthra, Alexandra~E Moylett, Dan~E. Browne, and Earl~T. Campbell.
\newblock ``Unlocking early fault-tolerant quantum computing with mitigated magic dilution''.
\newblock \href{https://dx.doi.org/10.1088/2058-9565/ae0aef}{Quantum Sci. Technol. {\bf 10}, 045066}~(2025).

\bibitem{siri2025deep}
Pallakonda Siri, Anitha G., and Maddukuri Reshma Naga~Venkata Durga.
\newblock ``A deep learning driven quantum hybrid approach for early breast cancer detection''.
\newblock In 5th Int. Conf. Artif. Intell. Signal Process.
\newblock \href{https://dx.doi.org/10.1109/AISP68263.2025.11396250}{Pages 1--5}.
\newblock IEEE~(2025).

\bibitem{saxena2024error}
Gaurav Saxena and Thi~Ha Kyaw.
\newblock ``Error mitigation by restricted evolution''~(2024).
\newblock  \href{http://arxiv.org/abs/2409.06636}{arXiv:2409.06636}.

\bibitem{Howard_2017}
Mark Howard and Earl Campbell.
\newblock ``Application of a resource theory for magic states to fault-tolerant quantum computing''.
\newblock \href{https://dx.doi.org/10.1103/physrevlett.118.090501}{Phys. Rev. Lett. {\bf 118}, 090501}~(2017).

\bibitem{aaronson:2004}
Scott Aaronson and Daniel Gottesman.
\newblock ``Improved simulation of stabilizer circuits''.
\newblock \href{https://dx.doi.org/10.1103/physreva.70.052328}{Phys. Rev. A {\bf 70}, 052328}~(2004).

\bibitem{MPS:2023}
Sangchul Oh and Sabre Kais.
\newblock ``Comparison of quantum advantage experiments using random circuit sampling''.
\newblock \href{https://dx.doi.org/10.1103/PhysRevA.107.022610}{Phys. Rev. A {\bf 107}, 022610}~(2023).

\bibitem{aziz2026classicalsimulationslowmagic}
Kemal Aziz, Haining Pan, Michael~J. Gullans, and J.~H. Pixley.
\newblock ``Classical simulations of low magic quantum dynamics''~(2026).
\newblock  \href{http://arxiv.org/abs/2508.20252}{arXiv:2508.20252}.

\bibitem{Seddon_2019}
James~R. Seddon and Earl~T. Campbell.
\newblock ``Quantifying magic for multi-qubit operations''.
\newblock \href{https://dx.doi.org/10.1098/rspa.2019.0251}{Proc. R. Soc. A {\bf 475}, 20190251}~(2019).

\bibitem{bravyi:2019}
Sergey Bravyi, Dan Browne, Padraic Calpin, Earl Campbell, David Gosset, and Mark Howard.
\newblock ``Simulation of quantum circuits by low-rank stabilizer decompositions''.
\newblock \href{https://dx.doi.org/10.22331/q-2019-09-02-181}{Quantum {\bf 3}, 181}~(2019).

\bibitem{SP:2017}
Ryan~S. Bennink, Erik~M. Ferragut, Travis~S. Humble, Jason~A. Laska, James~J. Nutaro, Mark~G. Pleszkoch, and Raphael~C. Pooser.
\newblock ``Unbiased simulation of near-{Clifford} quantum circuits''.
\newblock \href{https://dx.doi.org/10.1103/PhysRevA.95.062337}{Phys. Rev. A {\bf 95}, 062337}~(2017).

\bibitem{DFS:2020}
James~R. Seddon, Bartosz Regula, Hakop Pashayan, Yingkai Ouyang, and Earl~T. Campbell.
\newblock ``Quantifying quantum speedups: Improved classical simulation from tighter magic monotones''.
\newblock \href{https://dx.doi.org/10.1103/prxquantum.2.010345}{PRX Quantum {\bf 2}, 010345}~(2021).

\bibitem{DS:2022}
James~Robert Seddon.
\newblock ``Advancing classical simulators by measuring the magic of quantum computation''.
\newblock Phd thesis.
\newblock University College London.
\newblock ~(2022).
\newblock  url:~\url{https://discovery.ucl.ac.uk/id/eprint/10146361/}.

\bibitem{pashayan:2015}
Hakop Pashayan, Joel~J. Wallman, and Stephen~D. Bartlett.
\newblock ``Estimating outcome probabilities of quantum circuits using quasiprobabilities''.
\newblock \href{https://dx.doi.org/10.1103/PhysRevLett.115.070501}{Phys. Rev. Lett. {\bf 115}, 070501}~(2015).

\bibitem{hoeffding:1963}
Wassily Hoeffding.
\newblock ``Probability inequalities for sums of bounded random variables''.
\newblock \href{https://dx.doi.org/10.2307/2282952}{J. Am. Stat. Assoc. {\bf 58}, 13--30}~(1963).

\bibitem{welford:1962}
B.~P. Welford.
\newblock ``{Note on a method for calculating corrected sums of squares and products}''.
\newblock \href{https://dx.doi.org/10.1080/00401706.1962.10490022}{Technometrics {\bf 4}, 419--420}~(1962).

\bibitem{howard2021time}
Steven~R. Howard, Aaditya Ramdas, Jon McAuliffe, and Jasjeet Sekhon.
\newblock ``{Time-uniform, nonparametric, nonasymptotic confidence sequences}''.
\newblock \href{https://dx.doi.org/10.1214/20-AOS1991}{Ann. Stat. {\bf 49}, 1055--1080}~(2021).

\bibitem{maurer:2009}
Andreas Maurer and Massimiliano Pontil.
\newblock ``Empirical {Bernstein} bounds and sample-variance penalization''.
\newblock In Proc. 22nd Annu. Conf. Learn. Theory.
\newblock ~(2009).
\newblock  \href{http://arxiv.org/abs/0907.3740}{arXiv:0907.3740}.

\bibitem{Dankert}
Christoph Dankert, Richard Cleve, Joseph Emerson, and Etera Livine.
\newblock ``Exact and approximate unitary 2-designs and their application to fidelity estimation''.
\newblock \href{https://dx.doi.org/10.1103/PhysRevA.80.012304}{Phys. Rev. A {\bf 80}, 012304}~(2009).

\end{thebibliography}
\appendix
\clearpage

\section{Algorithm Details}\label{sec:algorithm-details}
This appendix provides bitwise implementation details used in Section~\ref{sec:sample-value}. 
Appendix~\ref{sec:CH-form} explains CH-form \cite{bravyi:2019} to define $\textsc{SampleValue}$ in Algorithm~\ref{alg:ESSENCE}. Appendix~\ref{sec:InnerProd} describes our inner product procedure using the CH-form. Appendix~\ref{sec:Uc-product} proves the C-type product rule, and Appendix~\ref{sec:binary-matrix-multiplication} presents the bitwise implementation.

Table~\ref{tab:subroutine-complexity} summarizes the time complexities of the subroutines and algorithms used in our work.

\begin{table*}[htbp]
    \centering
    \caption{Time complexities of algorithms that compute expectation values.}
    \label{tab:subroutine-complexity}
    \vspace{0.5em}
    \begin{tabular}{ll}
        \toprule
        Procedure & Time complexity \\
        \midrule
        \textsc{StateSampler} \cite{bravyi:2019} & $\mathcal{O}(n^4w)$ \\
        \textsc{BitInner} \cite{bravyi:2019} & $\mathcal{O}(n^2)$ \\
        \textsc{GetXi} \cite{bravyi:2019} & $\mathcal{O}(w)$ \\
        \textsc{ATransposeB} & $\mathcal{O}(n^3)$ \\
        \textsc{UcProduct} & $5 \cdot \textsc{ATransposeB} + \mathcal{O}(n^3)$ \\
        \textsc{InnerProduct} & $\textsc{UcProduct} + \mathcal{O}(n) \cdot \mathcal{O}(n^2) + \textsc{BitInner}$ \\
        \textsc{SampleValue}, projector $O$ & $2 \cdot \textsc{StateSampler} + 2\cdot\textsc{InnerProduct}$ \\
        \textsc{SampleValue}, other $O$ & $2 \cdot \textsc{StateSampler} + k\cdot[\textsc{InnerProduct} + \mathcal{O}(n^2)]$ \\        
        \textsc{ExpVal} & $h_{\mathrm{Hoeffding}}\cdot\textsc{SampleValue} + \textsc{GetXi}$ \\      
        \textsc{Welford} & $\kappa\cdot\textsc{SampleValue}$ \\
        \textsc{CarveExpVal} & $h_{\mathrm{CARVE}}\cdot\textsc{SampleValue}+ \textsc{GetXi}$ \\
        \bottomrule
    \end{tabular}
\end{table*}

\subsection{Sampling and Processing Stabilizer States in CH-form} \label{sec:CH-form}
For each $n$-qubit stabilizer state $\ket{\phi}$, there exists a CH-form representation $\ket{\phi} = \omega U_CU_H\ket{s}$, where $\omega \in \mathbb{C}$, $s \in \mathbb{Z}_{2}^n$, and $U_C$ and $U_H$ are C-type and H-type Clifford operators, respectively. Specifically, $U_C$ satisfies $U_C\ket{0}^{\otimes n} = \ket{0}^{\otimes n}$, and $U_H$ is a tensor product of Hadamard gates. For computational efficiency, $U_C$ is represented by the tuple $(F, G, M, \gamma)$, where $F$, $G$, and $M$ are $n\times n$ binary matrices, and $\gamma \in \mathbb{Z}_4^n$. They are defined as values satisfying
\begin{equation}
        U_C^\dagger Z_p U_C = \prod_{j=1}^n Z_j^{G_{p,j}},
        \qquad U_C^\dagger X_p U_C = i^{\gamma_p} \prod_{j=1}^n X_j^{F_{p,j}}Z_j^{M_{p,j}}.
\end{equation}
Similarly, $U_H$ is represented by the binary vector $v \in \mathbb{Z}_2^n$, defined as
\begin{equation}
    U_H = \prod_{j=1}^n H_j^{v_j}.
\end{equation}
In summary, a state in CH-form is stored and specified by the tuple $(F,G,M,\gamma,v,s,\omega)$, requiring $\mathcal{O}(n^2)$ memory.

Applying a general Clifford gate to the state requires $\mathcal{O}(n^4)$ time. This includes decomposing the gate into $O(n^2)$ basic blocks --- Pauli, Hadamard, S, and CNOT gates --- that takes up to $\mathcal{O}(n^2)$ time to update via a look-up table (\textsc{LookUp}). Based on these foundational operations, \textsc{StateSampler} that samples stabilizer states under Eq.~\eqref{eqn:sample-trajectory} requires $\mathcal{O}(n^4w)$ time to sample one state applying an $n$-qubit circuit with $w$ gates. This leaves at most $k$ iterations of the function \textsc{InnerProduct} to be characterized to define \textsc{SampleValue} as following.

\begin{algorithm}[H]
    \caption{One real-valued ESSENCE sample}
    \label{alg:sample-value}
    \begin{algorithmic}
    \Procedure{SampleValue}{$C,O$}
        \State $\varphi_1 \gets \textsc{StateSampler}(C)$
        \State $\varphi_2 \gets \textsc{StateSampler}(C)$
        \If {$O=\ket{\phi}\bra{\phi}$ is a stabilizer projector}
            \State $x \gets \mathfrak{Re}\left(\textsc{InnerProduct}(\varphi_1,\phi)\cdot \textsc{InnerProduct}(\phi,\varphi_2)\right)$
        \ElsIf {$O=\sum_{i=1}^{k}c_iP_i$}
            \State Initialize $x \gets 0$
            \For {$i = 1$ to $k$}
                \State $x \gets x + c_i \cdot \mathfrak{Re}\left(\textsc{InnerProduct}(\varphi_1,\textsc{LookUp}[P_i](\varphi_2))\right)$
            \EndFor
        \EndIf
        \State \Return $x$
    \EndProcedure
    \end{algorithmic}
\end{algorithm}

\subsection{Inner Product of Stabilizer States} \label{sec:InnerProd}
The inner product of stabilizer states $\ket{\phi_1}$ and $\ket{\phi_2}$,
\begin{equation}
    \braket{\phi_1|\phi_2} = \omega_1^*\omega_2\braket{s_1|U_H^{[1]}U_C^{[1]\dagger}U_C^{[2]}U_H^{[2]}|s_2},
\end{equation}
can be computed by first constructing the combined C-type operator $U_C = U_C^{[1]\dagger}U_C^{[2]}$ then applying Hadamard gates to the state $U_C U_H^{[2]}\ket{s_2}$ according to the H-type operator $U_H^{[1]}$, and finally evaluating the inner product (\textsc{BitInner}) between the bitstring $\ket{s_1}$ and the state $U_H^{[1]} U_C U_H^{[2]} \ket{s_2}$ with the method provided in \cite{bravyi:2019}. This is summarized in Algorithm~\ref{alg:inner-product}.

\begin{algorithm}[H]
    \caption{Inner product of stabilizer states in CH-form}
    \label{alg:inner-product}
    \begin{algorithmic}
        \Procedure{InnerProduct}{$\phi_1, \phi_2$}
        \State $(F^{[1]}, G^{[1]}, M^{[1]}, \gamma^{[1]}, v^{[1]}, s_1, \omega_1) \gets \phi_1$
        \State $(F^{[2]}, G^{[2]}, M^{[2]}, \gamma^{[2]}, v^{[2]}, s_2, \omega_2) \gets \phi_2$
        \State $(F, G, M, \gamma) \gets \textsc{UcProduct}((F^{[1]}, G^{[1]}, M^{[1]}, \gamma^{[1]}), (F^{[2]}, G^{[2]}, M^{[2]}, \gamma^{[2]}))$
        \State $\phi_{12} \gets (F, G, M, \gamma, v^{[2]}, s_2, \omega_2)$
        \For {$p=1$ to $n$}
        \If {$v_p^{[1]} = 1$}
        \State $\phi_{12} \gets \textsc{LookUp}[H_p](\phi_{12})$
        \EndIf
        \EndFor
        \State \Return $\omega_1^* \cdot \textsc{BitInner}(s_1, \phi_{12})$
        \EndProcedure
    \end{algorithmic}
\end{algorithm}
It takes $\mathcal{O}(n^3)$ time applying Hadamard gate and calculating the inner product with a bitstring, including the procedure to find the parameters $(F, G, M, \gamma)$ from \textsc{UcProduct} justified by the following Sections~\ref{sec:Uc-product} and \ref{sec:binary-matrix-multiplication}.

\subsection{Product and Inverse of C-type Matrices} \label{sec:Uc-product}
\subsubsection{Product of C-type Matrices}
For two C-type matrices $U_C^{[1]}$ and $U_C^{[2]}$ where each matrix is represented by $(F^{[1]}, G^{[1]}, M^{[1]}, \gamma^{[1]})$ and $(F^{[2]}, G^{[2]}, M^{[2]}, \gamma^{[2]})$, respectively, the values $(F, G, M, \gamma)$ that represents $U_C^{[1]}U_C^{[2]}$ are
\begin{equation}
    \begin{split}
        (F, G, M, \gamma) &= (F^{[1]}, G^{[1]}, M^{[1]}, \gamma^{[1]}) * (F^{[2]}, G^{[2]}, M^{[2]}, \gamma^{[2]})\\
        &= (F^{[1]}F^{[2]}, G^{[1]}G^{[2]}, F^{[1]}M^{[2]} + M^{[1]}G^{[2]}, \gamma^{[1]} + F^{[1]}\gamma^{[2]} + 2S)
    \end{split}
\end{equation}
where $S$ is the binary vector representing sign induced by anti-commutativity of $X$ and $Z$, which is equal to
\begin{equation}\label{eqn:sign-vector}
    S_p = \sum_{r=1}^n S_{pr} \bmod 2,
\end{equation}
an accumulated sum of values
\begin{equation}\label{eqn:sign-matrix}
    S_{pr} = \sum_{q=1}^n F^{[1]}_{pq}F^{[2]}_{qr} \sum_{q'=1}^{q-1} F_{pq'}^{[1]}M_{q'r}^{[2]}.
\end{equation}
It is because
\begin{equation}
    \begin{split}
        U_C^\dagger X_p U_C = &U_C^{[2]\dagger}U_C^{[1]\dagger}X_pU_C^{[1]}U_C^{[2]} = U_C^{[2]\dagger}\left(i^{\gamma_p^{[1]}}\prod_qX_q^{F_{pq}^{[1]}}Z_q^{M_{pq}^{[1]}}\right)U_C^{[2]}\\
        = &i^{\gamma_p^{[1]}}\prod_q \left(i^{F_{pq}^{[1]}\gamma_q^{[2]}}\prod_r X_r^{F_{pq}^{[1]}F_{qr}^{[2]}}Z_r^{F_{pq}^{[1]}M_{qr}^{[2]}}\right)\prod_q \left(\prod_{r}Z_r^{M_{pq}^{[1]}G_{qr}^{[2]}}\right)\\
        = &i^{\gamma_p^{[1]} + \sum_q F_{pq}^{[1]}\gamma_q^{[2]}}\cdot \prod_r \left((-1)^{S_{pr}}\cdot \prod_q X_r^{F_{pq}^{[1]}F_{qr}^{[2]}}Z_r^{F_{pq}^{[1]}M_{qr}^{[2]}}\right)\prod_r \left(Z_r^{\sum_q M_{pq}^{[1]}G_{qr}^{[2]}}\right)\\
        = &i^{\gamma_p^{[1]} + \sum_q F_{pq}^{[1]}\gamma_q^{[2]}}\cdot (-1)^{\sum_r{S_{pr}}}\cdot\prod_r \left(\prod_q X_r^{F_{pq}^{[1]}F_{qr}^{[2]}}Z_r^{F_{pq}^{[1]}M_{qr}^{[2]}}\right)\prod_r \left(Z_r^{\sum_q M_{pq}^{[1]}G_{qr}^{[2]}}\right)\\
        = & i^{\gamma_p^{[1]} + (F^{[1]}\gamma^{[2]})_p +2S_{p}} \cdot \prod_r X_r^{(F^{[1]}F^{[2]})_{pr}}Z_r^{(F^{[1]}M^{[2]})_{pr} + (M^{[1]}G^{[2]})_{pr}}
    \end{split}
\end{equation}
and
\begin{equation}
    \begin{split}
        U_C^\dagger Z_p U_C = &U_C^{[2]\dagger}U_C^{[1]\dagger}Z_pU_C^{[1]}U_C^{[2]} = U_C^{[2]\dagger}\left(\prod_qZ_q^{G_{pq}^{[1]}}\right)U_C^{[2]}\\
        = &\prod_r Z_r^{\sum_q G_{pq}^{[1]}G_{qr}^{[2]}}\\
        = &\prod_rZ_r^{(G^{[1]}G^{[2]})_{pr}}.
    \end{split}
\end{equation}

\subsubsection{Inverse of C-type Matrices}
For a C-type matrix $U_C$ specified by $(F, G, M, \gamma)$, the inverse $U_C^\dagger$ would have $(\tilde{F}, \tilde{G}, \tilde{M}, \tilde{\gamma})$ where 
\begin{equation}
    (\tilde{F}, \tilde{G}, \tilde{M}, \tilde{\gamma}) * (F, G, M, \gamma) = (\tilde{F}F, \tilde{G}G, \tilde{F}M + \tilde{M}G, \tilde{\gamma} + \tilde{F}\gamma + 2\tilde{S}) = (I, I, O, 0^n)
\end{equation}
which results in
\begin{equation}\label{eqn:Uc-inv}
    (F, G, M, \gamma)^{-1} = (F^{-1}, G^{-1}, F^{-1}MG^{-1}, -F^{-1}\gamma + 2\tilde{S}) = (G^\top, F^\top, G^\top MF^\top, -G^\top\gamma + 2\tilde{S})
\end{equation}
using the relation $FG^\top \equiv I \pmod{2}$ by its definition. Here, $\tilde{S}$ is obtained from Eq.~\eqref{eqn:sign-vector} via the following relation:
\begin{equation}\label{eqn:sign-inv}
    (G^\top, F^\top, G^\top MF^\top, -G^\top\gamma) * (F, G, M, \gamma) = (I, I, O, 2\tilde{S})
\end{equation}
or equivalently, by combining Eq.~\eqref{eqn:sign-vector} with Eq.~\eqref{eqn:Uc-inv} as follows:
\begin{equation} \label{eqn:sign-vector-inv}
    \tilde{S}_p = \sum_{r=1}^n \sum_{q=1}^n \tilde{F}_{pq}F_{qr} \sum_{q'=1}^{q-1}\tilde{F}_{pq'}M_{q'r} = \sum_{r=1}^n \sum_{q=1}^n G_{qp}F_{qr} \sum_{q'=1}^{q-1}G_{q'p}M_{q'r}
\end{equation}

Finally, the parameters $(F, G, M, \gamma) = (F^{[1]}, G^{[1]}, M^{[1]}, \gamma^{[1]})^{-1} * (F^{[2]}, G^{[2]}, M^{[2]}, \gamma^{[2]})$ representing $U_C^{[1]\dagger} U_C^{[2]}$ are given by
\begin{equation}
    \begin{cases}F = G^{[1]\top}F^{[2]}\\
    G = F^{[1]\top}G^{[2]}\\
    M = G^{[1]\top}M^{[2]} + G^{[1]\top}M^{[1]}F^{[1]\top}G^{[2]}\\
    \gamma = G^{[1]\top}(\gamma^{[2]} - \gamma^{[1]}) + 2S\end{cases}
\end{equation}
where the components of the binary vector $S$ are evaluated as
\begin{equation}
    S_p = \sum_{q=1}^n G^{[1]}_{qp}\sum_{r=1}^n \left(F^{[1]}_{qr} \sum_{q'=1}^{q-1}G^{[1]}_{q'p}M^{[1]}_{q'r} + F^{[2]}_{qr} \sum_{q'=1}^{q-1}G^{[1]}_{q'p}M^{[2]}_{q'r}\right).
\end{equation}

\subsection{Efficient Calculation of Sign Vector with bitwise AND Operation} \label{sec:binary-matrix-multiplication}
\subsubsection{Binary Matrix Product}
Previous work stores a binary matrix as a set of row vectors (or an $n$-bit value representing each row), to simplify a specific type of binary matrix product $A^\top B$ to $A \circ B$, as
\begin{equation}
    (A \circ B)_{ij} = (A^\top B)_{ij} = \sum_p A_{pi}B_{pj} = A_{\cdot i} \cdot B_{\cdot j} = |A_{\cdot i}\ \&\ B_{\cdot j}|_H \bmod 2
\end{equation}
reducing it to $\mathcal{O}(n^2)$ bitwise AND operations on $n$-bit values.
\begin{algorithm}[ht]
    \caption{Efficient binary matrix product}
    \begin{algorithmic}
        \Procedure{ATransposeB}{$A, B$}
        \State $M \gets 0^{n\times n}$
        \For {$p=1$ to $n$}
        \For {$q=1$ to $n$}
        \State $M_{pq} \gets |A_{\cdot p}\ \&\ B_{\cdot q}|_H \bmod 2$
        \EndFor
        \EndFor
        \State \Return $M$
        \EndProcedure
    \end{algorithmic}
\end{algorithm}

Using this technique, the values $(F, G, M, \gamma) = (F^{[1]}, G^{[1]}, M^{[1]}, \gamma^{[1]})^{-1} * (F^{[2]}, G^{[2]}, M^{[2]}, \gamma^{[2]})$ that represents $U_C^{[1]\dagger} U_C^{[2]}$ are calculated by
\begin{equation}
    \begin{cases}F = G^{[1]\top}F^{[2]} = G^{[1]}\circ F^{[2]}\\
    G = F^{[1]\top}G^{[2]} = F^{[1]}\circ G^{[2]}\\
     M = G^{[1]\top}M^{[2]} + G^{[1]\top}M^{ [1]}F^{[1]\top}G^{[2]} = G^{[1]} \circ M^{[2]} + G^{[1]} \circ (MT^{[1]} \circ G)\\
    \gamma = -G^{[1]\top}\gamma^{[1]} + 2\tilde{S}^{[1]} + G^{[1]\top}\gamma^{[2]} + 2S^{[1,2]} = G^{[1]} \circ (\gamma^{[2]} - \gamma^{[1]}) + 2(\tilde{S}^{[1]} + S^{[1,2]})\end{cases}
\end{equation}
where the transposed matrices such as $MT = M^\top$, $FT = F^\top$ are precalculated for desired matrix multiplication which takes $\mathcal{O}(n^2)$ bitwise AND operations on $n$-bit values. 

\subsubsection{Sign Vector Accumulation}
Substituting Eq.~\eqref{eqn:Uc-inv} to Eq.~\eqref{eqn:sign-vector} and combining with Eq.~\eqref{eqn:sign-vector-inv} gives $\tilde{S}^{[1]} + S^{[1,2]}$ as
\begin{equation}\label{eqn:sign-vector-sum}
    S_p = \tilde{S}^{[1]}_p + S^{[1,2]}_p = \sum_{q=1}^n G^{[1]}_{qp}\sum_{r=1}^n \left(F^{[1]}_{qr} \sum_{q'=1}^{q-1}G^{[1]}_{q'p}M^{[1]}_{q'r} + F^{[2]}_{qr} \sum_{q'=1}^{q-1}G^{[1]}_{q'p}M^{[2]}_{q'r}\right).
\end{equation}
To calculate $S_p$ using Eq.~\eqref{eqn:sign-vector-sum} with binary matrix multiplication for each $p$, define binary vectors $C^{[i]}_{\cdot p}(q)$ for $i = 1, 2$ which is defined recursively with $C^{[i]}_{\cdot p}(1) = 0^n$ as
\begin{equation}
    C_{rp}^{[i]}(q+1) = \sum_{q'=1}^{q}G^{[1]}_{q'p}M^{[i]}_{q'r} = G^{[1]}_{q,p}MT^{[i]}_{r,q} + C_{rp}^{[i]}(q).
\end{equation}
Thus one can also rewrite Eq.~\eqref{eqn:sign-vector-sum} by
\begin{equation}
    \sum_{q=1}^n G^{[1]}_{qp}\sum_{r=1}^n \left(F^{[1]}_{qr} C^{[1]}_{rp}(q) + F^{[2]}_{qr} C^{[2]}_{rp}(q)\right) = \sum_{q=1}^n G^{[1]}_{qp} \left(FT^{[1]}_{\cdot q} \cdot C^{[1]}_{\cdot p}(q) + FT^{[2]}_{\cdot q} \cdot C^{[2]}_{\cdot p}(q) \right).
\end{equation}
This leaves $\mathcal{O}(n^2)$ bitwise-and operations on $n$-bit values cumulating $C^{[1]}$ and $C^{[2]}$ through every $q$ for each $p$.
\begin{algorithm}[ht]
    \caption{Product of $U_C$}
    \begin{algorithmic}
        \Procedure{UcProduct}{$(F^{[1]}, G^{[1]}, M^{[1]}, \gamma^{[1]}), (F^{[2]}, G^{[2]}, M^{[2]}, \gamma^{[2]})$}
        \State $F \gets \textsc{ATransposeB}(G^{[1]}, F^{[2]})$
        \State $G \gets \textsc{ATransposeB}(F^{[1]}, G^{[2]})$
        \State $M1 \gets \textsc{ATransposeB}(G^{[1]}, M^{[2]})$
        \State $M2 \gets \textsc{ATransposeB}(G^{[1]}, \textsc{ATransposeB}(MT^{[1]}, G))$
        \State $\gamma_p \gets 0^n$
        \For {$p=1$ to $n$}
        \For {$q=1$ to $n$}
        \State $\gamma_p \gets \gamma_p +_4 G_{qp}^{[1]}(\gamma_q^{[2]} -_4 \gamma_q^{[1]})$
        \State $M_{pq} \gets M1_{pq} \oplus M2_{pq}$
        \EndFor
        \EndFor
        \For {$p=1$ to $n$}
        \State $C^{[1]} \gets 0^n$, $C^{[2]} \gets 0^n$
        \For {$q=1$ to $n$}
        \If {$G_{qp}^{[1]} = 1$}
        \If {$|C^{[1]}\ \&\ FT^{[1]}_{\cdot q}|_H \neq |C^{[2]}\ \&\ FT^{[2]}_{\cdot q}|_H \bmod 2$}
        \State $S_p \gets S_p \oplus 1$
        \EndIf
        \State $C^{[1]} \gets C^{[1]} \oplus MT_{\cdot q}^{[1]}$
        \State $C^{[2]} \gets C^{[2]} \oplus MT_{\cdot q}^{[2]}$
        \EndIf
        \EndFor
        \EndFor
        \State $\gamma \gets \gamma +_4 2S$
        \State \Return $(F, G, M, \gamma)$
        \EndProcedure
    \end{algorithmic}
\end{algorithm}

\clearpage
\section{Analysis on CARVE protocol} \label{sec:variance-estimation-appendix}
This appendix justifies the protocol and proves Theorem~\ref{thm:variance-estimation} and Theorem~\ref{thm:nisq-thm}.

\subsection{Bounding Number of Samples} \label{sec:CARVE-proof-appendix}
We first state and prove a stronger version of Bernstein's inequality.
\begin{lemma} [time-uniform Bernstein's inequality]
    Let $X_1, X_2, \cdots$ be a sequence of i.i.d. random variables, which satisfies $\mathbb{E}[X_i] = 0$, $\text{Var}(X_i)=\sigma^2$, and $\abs{X_i} \le 1$. Let $\bar{X}_t = \sum_{i=1}^{t}{X_i}/t$ be a sampled mean upon $t$. Then, for any $\epsilon > 0$, the following inequality holds:
    \begin{equation}\label{eq:time-uniform-bernstein}
        \mathbb{P}\left[\exists T \ge t : \abs{\bar{X}_T} > \epsilon \right] \le 2e^{-t\epsilon^2/(2\sigma^2 + 2\epsilon/3)}
    \end{equation}
\end{lemma}

\begin{proof}
    To establish the time-uniform bound over all $T \ge t$, we utilize the framework of $\psi$-yielding supermartingales formalized by Howard et al. (2020) \cite{howard2021time}. Using the Taylor series expansion for $\lambda \in [0, 3)$, we obtain:
    \begin{equation}
        \mathbb{E}[e^{\lambda X_i}] = 1 + \lambda \mathbb{E}[X_i] + \sum_{k=2}^{\infty}{\frac{\lambda^k\mathbb{E}[X_i^k]}{k!}}.
    \end{equation}
    As $\mathbb{E}[X_i] = 0$ and $\abs{X_i} \le 1$, the inequality $\mathbb{E}[X_i^k] \le \mathbb{E}[X_i^2] = \sigma^2$ holds for every $k \ge 2$ making
    \begin{equation}
        \begin{split}
            \mathbb{E}[e^{\lambda X_i}] \le 1 + \sum_{k=2}^{\infty}{\frac{\lambda^k\sigma^2}{k!}} \le 1 + \sum_{k=2}^{\infty}{\frac{\lambda^k\sigma^2}{2 \cdot 3^{k-2}}} = 1 + \frac{\lambda^2 \sigma^2}{2}\cdot \frac{1}{1-\lambda/3} \le \exp\left(\frac{\lambda^2\sigma^2}{2(1-\lambda/3)}\right).
        \end{split}
    \end{equation}
    Therefore,
    \begin{equation}
        \mathbb{E}\left[e^{\lambda X_i - \psi(\lambda)}\right] \le 1
    \end{equation}
    for $0 \le \lambda < 3$ where
    \begin{equation}
        \psi(\lambda) =  \frac{\lambda^2 \sigma^2 / 2}{1 - \lambda/3}.
    \end{equation}

    Let $S_T = \sum_{i=1}^T X_i$ be the partial sum and $\mathcal{F}_T = \sigma(X_1, X_2, \dots, X_T)$ be the natural filtration generated by the sequence, with $\mathcal{F}_0$ being the trivial $\sigma$-algebra. Then the discrete stochastic process 
    \begin{equation}
        M_T(\lambda) = \exp(\lambda S_T - T\psi(\lambda))
    \end{equation}
    is a non-negative supermartingale, as $M_0(\lambda) = 1$ and
    \begin{equation}
        \mathbb{E}[M_{T+1}(\lambda) \mid \mathcal{F}_T] = M_T(\lambda) \cdot \mathbb{E}[\exp(\lambda X_{T+1} - \psi(\lambda))] \le M_T(\lambda)
    \end{equation}
    by observing the conditional expectation.

    We now analyze the probability of the sample mean exceeding $\epsilon$. The event $\{\exists T \ge t : \bar{X}_T > \epsilon\}$ is algebraically equivalent to:
    \begin{equation}
        \{\exists T \ge t : \lambda S_T - T\psi(\lambda) > T(\lambda \epsilon - \psi(\lambda))\}.
    \end{equation}
    If we restrict $\lambda$ such that $\lambda\epsilon > \psi(\lambda)$, then the right-hand side is strictly positive and increasing with $T$. This allows us to establish the following subset relation:
    \begin{equation}
        \{\exists T \ge t : \bar{X}_T > \epsilon\} \subset \left\{\sup_{T \ge t} M_T(\lambda) > \exp(t(\lambda\epsilon - \psi(\lambda)))\right\}.
    \end{equation}
    Applying Ville's inequality to the non-negative supermartingale $M_T$, we get:
    \begin{equation}\label{eqn:ville-inequality}
        \mathbb{P}(\exists T \ge t : \bar{X}_T > \epsilon) \le \frac{\mathbb{E}[M_0]}{\exp\left(t(\lambda\epsilon - \psi(\lambda))\right)} = \exp\left(-t(\lambda\epsilon - \psi(\lambda))\right).
    \end{equation}
    
    We substitute the following value for $\lambda$:
    \begin{equation}
        \lambda = \frac{\epsilon}{\sigma^2 + \epsilon/3}
    \end{equation}
    yielding the corresponding value for $\psi(\lambda)$:
    \begin{equation}
        \psi(\lambda) = \frac{\epsilon^2}{2(\sigma^2+\epsilon/3)}.
    \end{equation}
    These satisfy both $0 < \lambda < 3$ and $\lambda\epsilon > \psi(\lambda)$. Substituting these into Eq.~\eqref{eqn:ville-inequality}, we obtain the one-sided inequality:
    \begin{equation}
        \mathbb{P}(\exists T \ge t : \bar{X}_T > \epsilon) \le \exp\left( -\frac{t\epsilon^2}{2\sigma^2 + 2\epsilon/3} \right)
    \end{equation}
    while the other side can be treated by the same logic applied to $Y_T = -X_T$ since $Y_T$ satisfies the exact same conditions as $X_T$. Therefore,
    \begin{equation}
        \mathbb{P}\left(\exists T \ge t : |\bar{X}_T| > \epsilon \right) \le 2\exp\left(-\frac{t\epsilon^2}{2\sigma^2 + 2\epsilon/3}\right).
    \end{equation}
\end{proof}

Now that the bounds are established, we can use the lemma to justify CARVE.

\begin{proposition}
    \textnormal{\textsc{CarveExpVal}} solves EVE$(\epsilon, \delta)$.
\end{proposition}
\begin{proof}
Let $X_1,X_2,\cdots$ denote the real-valued outputs of 
$\textsc{SampleValue}(C,O)$. 
All random variables satisfy:
\begin{equation}
    \mathbb{E}[\xi_\star X_i]=\mu_{C,O},\qquad 
    |X_i|\le \|O\|,\qquad 
    |\xi_\star X_i-\mu_{C,O}|\le (\xi_\star+1)\|O\|,
\end{equation}
where $\mu_{C,O}=\bra{0^{\otimes n}}C^\dagger O C\ket{0^{\otimes n}}$ is the true expectation value of the circuit $C$.

By Theorem~10 of Maurer et al. \cite{maurer:2009}, if we estimate the empirical variance $s^2$ using $\kappa$ samples of a random variable bounded within $[-\|O\|, \|O\|]$, the true variance $\sigma^2$ satisfies the following inequality:
\begin{equation} \label{eq:var-ineq}
    \mathbb{P}\left[\sigma > s + \|O\| \cdot \sqrt{\frac{8\log(\delta^{-1})}{\kappa-1}}\right] \le \delta.
\end{equation}

Suppose that a lower bound $0 < L_O \le \|O\|$ is given. CARVE measures $\kappa$ samples to compute an empirical variance $s^2$ of \textsc{SampleValue}. Using this, it defines $\sigma_M$ and $h_\mathrm{CARVE}$:
\begin{equation} 
    \sigma_M = \min\left(1, \frac{s}{L_O} +  \sqrt{\frac{8\log(2\delta^{-1})}{\kappa-1}}\right),
\end{equation}
\begin{equation}
    h_\mathrm{CARVE} = \max\left(\kappa, \left\lceil 2\left(\xi_\star^2\sigma_M^2+\frac{(\xi_\star + 1)\epsilon}{3}\right)\epsilon^{-2}\log(4\delta^{-1}) \right\rceil\right)
\end{equation}

We must now verify the guarantee of reusing these samples. 
Let $N_\sigma$ be the required number of samples given the true variance $\sigma^2$ of \textsc{SampleValue}:
\begin{equation}
    N_\sigma = \lceil 2(\xi_\star^2 \sigma^2/\|O\|^2 + (\xi_\star + 1)\epsilon/3)\epsilon^{-2}\log(4\delta^{-1})\rceil
\end{equation}
and define two success events
\begin{equation}
    A=\{\sigma/\|O\| \le \sigma_M\},\qquad
    B=\left\{\forall n\ge N_\sigma:\abs{\xi_\star \bar X_n-\mu_{C,O}}\le \epsilon\|O\|\right\}.
\end{equation}
Equation \eqref{eq:var-ineq} gives $\mathbb{P}(A^C)\le \delta/2$, and Eq.~\eqref{eq:time-uniform-bernstein} gives $\mathbb{P}(B^C)\le \delta/2$. If the event $A \cap B$ occur, then $h_\mathrm{CARVE} \ge N_\sigma$, which guarantees that $|\xi_\star \bar{X}_{h_\mathrm{CARVE}} - \mu_{C,O}| \le \epsilon\|O\|$. The probability that the error exceeds $\epsilon$ is bounded by the probability that either $A$ or $B$ fails:
\begin{equation}
    \mathbb{P}\left[\abs{\xi_\star \bar X_{h_\mathrm{CARVE}}-\mu_{C,O}}>\epsilon\|O\|\right]
    \le \mathbb{P}\left[(A\cap B)^C\right]
    \le \mathbb{P}(A^C)+\mathbb{P}(B^C)
    \le \delta.
\end{equation}
Thus, the required accuracy and confidence bounds are satisfied.
\end{proof}

\subsection{Efficiency Analysis via Haar-random Circuits} \label{sec:haar-CARVE}

\varianceThm*
\begin{proof}
    We first observe that drawing a Haar random unitary $U$ is mathematically equivalent to the process below:
    \begin{enumerate}
        \item Draw a Haar-random unitary $V \sim \text{Haar}$.
        \item Draw a uniformly random Clifford gate $C \sim \mathcal{C}_n$.
        \item Set $U = CV$.
    \end{enumerate}

    As the Haar measure is left-invariant, the product $CV$ remains Haar-random. This invariance allows us to rewrite the expectation value of any function $f$ as:
    \begin{equation}
        \mathbb{E}_{U \sim \text{Haar}} [ f(U) ] = \mathbb{E}_{V \sim \text{Haar}} \left[ \mathbb{E}_{C \sim \mathcal{C}_n} [ f(CV) ] \right]
    \end{equation}
    Let $V = \sum_i \alpha_i K_i$ be the optimal Clifford decomposition of $V$. We are given the following variance bound:
    \begin{equation}
        \sigma^2(V) \le \sum_{i, j} p_i p_j |\bra{0}  K_i^\dagger O K_j \ket{0} |^2
    \end{equation}
    Let us define $f(V)$ as the summation term in the inequality:
    \begin{equation}
        f(V) = \sum_{i, j} p_i p_j |\bra{0} K_i^\dagger O K_j \ket{0}|^2 = \sum_{i, j} p_i p_j |\bra{\phi_i}O\ket{\phi_j}|^2
    \end{equation}
    where $\ket{\phi_i} = K_i \ket{0}$ and $\ket{\phi_j} = K_j \ket{0}$.
    
    Since $C \in \mathcal{C}_n$ is a Clifford operator, the product $CV$ has the decomposition $CV = \sum_i \alpha_i (CK_i)$. The probability coefficients $p_i$ and $p_j$ remain invariant under this global Clifford multiplication. Thus, the inner expectation over $\mathcal{C}_n$ becomes:
    \begin{equation}
        \mathbb{E}_{C \sim \mathcal{C}_n} [f(CV)] = \sum_{i, j} p_i p_j \mathbb{E}_{C \sim \mathcal{C}_n} \left[ \abs{\bra{\phi_i} C^\dagger O C \ket{\phi_j}}^2 \right]
    \end{equation}
    Assuming $O$ is a non-identity Pauli operator, conjugating it by a uniformly random Clifford $C$ maps it uniformly across the set of all $4^n - 1$ non-identity Pauli operators. Therefore, the expectation over $C$ simplifies to an average over the set of all non-identity Paulis $P \neq I$:
    \begin{equation}
        \mathbb{E}_{C \sim \mathcal{C}_n} \left[ \abs{\bra{\phi_i} C^\dagger O C \ket{\phi_j}}^2 \right] = \frac{1}{4^n - 1} \sum_{P \neq I} \abs{\bra{\phi_i} P \ket{\phi_j}}^2
    \end{equation}
    Using the Pauli operator completeness relation $\sum_P P \ket{\phi_j}\bra{\phi_j} P = 2^n I$, we can evaluate the sum over all Paulis:
    \begin{equation}
        \sum_{P} \abs{\bra{\phi_i} P \ket{\phi_j}}^2 = \sum_{P} \bra{\phi_i} P \ket{\phi_j}\bra{\phi_j} P \ket{\phi_i} = 2^n \braket{\phi_i | \phi_i} = 2^n
    \end{equation}
    Separating the identity term gives a strict upper bound:
    \begin{equation}
        \sum_{P \neq I} \abs{\bra{\phi_i} P \ket{\phi_j}}^2 = 2^n - \abs{\braket{\phi_i | \phi_j}}^2 \le 2^n
    \end{equation}
    Substituting this result back into the expectation value, and recognizing that $\sum_{i, j} p_i p_j = 1$, we obtain:
    \begin{equation}
        \mathbb{E}_{C \sim \mathcal{C}_n} [f(CV)] \le \sum_{i, j} p_i p_j \left( \frac{1}{4^n - 1} \sum_{P \neq I} \abs{\bra{\phi_i} P \ket{\phi_j}}^2 \right) \le \frac{2^n}{4^n - 1}   
    \end{equation}
    As this bound is a constant strictly independent of $V$, taking the outer expectation over $V \sim \text{Haar}$ leaves the bound unchanged, which completes the proof.
    
    \textit{Remark}: This application of Clifford twirling draws conceptual inspiration from the unitary 2-design properties outlined by Dankert et al. \cite{Dankert}.
\end{proof}

\subsection{Efficiency Analysis via Clifford + T circuits}\label{sec:clifford-t-CARVE}
In this section, we analyze the relative variance under the circuit model of Clifford + T gates. We begin by establishing a fundamental connection between the algorithm's relative variance and the purity of the sampled stabilizer branches ensemble.

\begin{lemma}\label{lem:purity-bound}
    Suppose the ESSENCE for an observable $O$ samples states $\ket{\phi_i}$ with probability $p_i$ via the subprocedure \textnormal{\textsc{StateSampler}}. The relative variance $\eta^2 = \sigma^2/\|O\|^2$ satisfies
    \begin{equation}
        \eta^2\le \Tr(\rho_b^2), \qquad \rho_b=\sum_ip_i\ket{\phi_i}\bra{\phi_i}.
    \end{equation}
    That is, the relative variance of ESSENCE is upper bounded by purity of branch ensemble.
\end{lemma}
\begin{proof}
    Let
    \begin{equation}
        \bar{O} = \frac{O}{\|O\|}
    \end{equation}
    be a normalized observable.
    Suppose the full $n$-qubit circuit $C$ is decomposed as 
    \begin{equation}
        C=\sqrt{\xi_\star}\sum_{i=1}^\chi{p_iV_i},\quad \text{where} \quad \sum_{i=1}^\chi{p_i}=1,\ p_i \ge 0,\ V_i\in \mathcal{C}_n.
    \end{equation}
    Here we absorbed global phases into $V_i$. Let's denote $\ket{\phi_i}=V_i\ket{0^n}$.
    
    Let $X$ be our estimator, which is a random outcome of \textsc{SampleValue}. Then,
    \begin{equation}
        \mathbb{P}(X= \mathfrak{Re}(\bra{\phi_i}O\ket{\phi_j})) = p_ip_j
    \end{equation}
    for every $i$ and $j$ from 1 to $\chi$, respectively. Thus, the relative variance $\eta^2 = \sigma^2/\|O\|^2$ satisfies
    \begin{equation}
        \eta^2 = \frac{\sigma^2}{\|O\|^2} \le
        \frac{\mathbb{E}[X^2]}{\|O\|^2} \le
        \mathbb{E}\left[\abs{\bra{\phi_i}\bar{O}\ket{\phi_j}}^2\right] =
        \sum_{i=1}^\chi\sum_{j=1}^\chi{p_ip_j\abs{\bra{\phi_i}\bar{O}\ket{\phi_j}}^2}.
    \end{equation}

    Using the definition of the branch ensemble $\rho_b$ from the lemma statement, we have
    \begin{equation}
        \sum_{i=1}^\chi\sum_{j=1}^\chi{p_ip_j\abs{\bra{\phi_i}\bar{O}\ket{\phi_j}}^2} = \Tr(\rho_b\bar{O}\rho_b\bar{O}) \le \Tr(\rho_b^2)
    \end{equation}
    where the final inequality follows because $\|\bar{O}\| = 1$.
\end{proof}

The bound established in Lemma~\ref{lem:purity-bound} motivates us to track how the purity $\Tr(\rho_b^2)$ degrades as the circuit is applied. By treating the probabilistic selection of stabilizer branches as a noisy quantum channel, we can model this degradation directly.  Specifically, the probabilistic decomposition of a $T$ gate---when conjugated by the surrounding Clifford circuit—maps to a channel defined by a generic $n$---qubit Pauli operator $P$. Thus, evaluating the ensemble's purity requires analyzing this generic Pauli channel. To formalize this, let $\mathcal{P}_n$ denote the set of $n$-qubit Pauli strings omitting global phases.

\begin{lemma}\label{lem:channel-purity}
    For any $n$-qubit Pauli operator $P \in \mathcal{P}_n$, define the quantum channel $\mathcal{E}_P$ as
    \begin{equation}
        \mathcal{E}_P(\rho) = \frac{1}{2} \rho + \frac{1}{2} \left(\frac{I - iP}{\sqrt{2}}\right) \rho \left(\frac{I + iP}{\sqrt{2}}\right).
    \end{equation}
    If $P$ is sampled uniformly from the set of non-identity $n$-qubit Pauli operators, then the expected purity of the output state $\mathcal{E}_P(\rho)$ for any $n$-qubit state $\rho$ is strictly bounded above by
    \begin{equation}
        \frac{1}{2^n} + \frac{3}{4}\left(\Tr(\rho^2) - \frac{1}{2^n}\right).
    \end{equation}
\end{lemma}
\begin{proof}
    Any $n$-qubit density matrix $\rho$ can be expressed via its Pauli decomposition
    \begin{equation}
        \rho = \frac{1}{2^n}\sum_{P \in \mathcal{P}_n}{c_P P}, \qquad c_P = \Tr(\rho P) \in \mathbb{R}.
    \end{equation}
    Consequently, the purity of the density matrix can be written simply as the sum of its squared coefficients:
    \begin{equation}
        \Tr(\rho^2)= \frac{1}{2^n}\sum_{P \in \mathcal{P}_n}{c_P^2}.
    \end{equation}
    
    We can decompose $\rho$ into parts that commute ($\rho_+$) and anti-commute ($\rho_-$) with $P$: 
    \begin{equation}
        \rho_+ = \frac{1}{2^n}\sum_{Q: [P, Q] = 0} c_QQ, \qquad \rho_- = \frac{1}{2^n}\sum_{Q: \{P, Q\} = 0} c_QQ.
    \end{equation}
    Then the action of the channel $\mathcal{E}_P$ yields
    \begin{equation}
        \mathcal{E}_P(\rho_+) =  \rho_+, \qquad
        \mathcal{E}_P(\rho_-) = \frac{I-iP}{2}\rho_-. 
    \end{equation}

    Squaring these outcomes gives:
    \begin{equation}\label{eq:rho-plus}
         (\mathcal{E}_P(\rho_+))^2 = \rho_+^2, 
    \end{equation}
    \begin{equation}\label{eq:rho-minus}
        (\mathcal{E}_P(\rho_-))^2 = \frac{1}{4}(\rho_-^2 - i\rho_-P\rho_- -iP\rho_-^2 - P\rho_-P\rho_-) = \frac{1}{2}\rho_-^2.
    \end{equation}
    Therefore, the purity is given as:
    \begin{equation}
        \Tr((\mathcal{E}_P(\rho))^2) = \frac{1}{2^n}\left(\sum_{Q : [P,Q]=0}{c_Q^2} + \frac{1}{2}\sum_{Q : \{P,Q\}=0}{c_Q^2}\right).
    \end{equation}
    The expected purity is:
    \begin{equation}
        \mathbb{E}_{P \sim \mathcal{P}_n \setminus \{I\}}\left[\Tr((\mathcal{E}_P(\rho))^2)\right] = \frac{1}{2^n} \sum_{Q \in \mathcal{P}_n}c_Q^2\left(\mathbb{P}[[P,Q]=0] + \frac{1}{2}\mathbb{P}[\{P,Q\}=0]\right).
    \end{equation}
    Under the uniform distribution $P \sim \mathcal{P}_n \setminus \{I\}$, the probabilities for commutation and anti-commutation are
    \begin{equation}
        \mathbb{P}[[P,Q]=0] = \begin{cases}
            1 & Q = I\\
            \dfrac{2^{2n-1}-1}{4^n-1} & Q \neq I
        \end{cases}, \qquad \mathbb{P}[\{P,Q\}=0]= \begin{cases}
            0 & Q=I \\
            \dfrac{2^{2n-1}}{4^n-1} & Q \neq I
        \end{cases}.
    \end{equation}
    Substituting these probabilities into the expectation gives
    \begin{equation}
        \mathbb{E}_{P\sim\mathcal{P}_n\setminus\set{I}}[\Tr((\mathcal{E}_P(\rho))^2)] =
        \frac{1}{2^n} \left( c_I^2 + \sum_{Q \in \mathcal{P}_n \setminus \set{I}}{c_Q^2 \frac{3\cdot 4^{n-1}-1}{4^n-1}}\right).
    \end{equation}
    Since $c_I = \Tr(\rho) = 1$ is fixed, we can simplify this expression to
    \begin{equation}
        \mathbb{E}_{P\sim\mathcal{P}_n\setminus\set{I}}[\Tr((\mathcal{E}_P(\rho))^2)] =
        \frac{1}{2^n} + \left(\Tr(\rho^2)-\frac{1}{2^n}\right)\frac{3-4^{1-n}}{4-4^{1-n}}.
    \end{equation}
    The lemma immediately follows.
\end{proof}

With the single-step purity bound established in Lemma~\ref{lem:channel-purity}, we are now equipped to bound the relative variance of the full Clifford + $T$ circuit. 

\varianceNISQthm*
\begin{proof}
    By Lemma~\ref{lem:purity-bound}, the relative variance of the estimator is bounded above by the purity of the branch ensemble, $\Tr(\rho_b^2)$. We proceed by evaluating the expected purity of the ensemble under the distribution of the given circuits.
    
    Up to a global phase, the optimal decomposition for the $T_k$-gate is
    \begin{equation}
        T_k= \frac{1}{\sqrt{2+\sqrt{2}}}(I + e^{-i\pi/4}S_k).
    \end{equation}
    Applying this $T_k$-gate probabilistically is equivalent to applying a quantum channel that mixes the identity and $S_k (\cdot) S_k^\dagger$ with equal probability. Because $S_k = (I - iZ_k)/\sqrt{2}$ up to a global phase, this operation corresponds exactly to the channel $\mathcal{E}_{Z_k}$ defined in Lemma \ref{lem:channel-purity}. This allows us to rewrite the entire circuit $C$ as a composite quantum channel $\mathcal{E}_C$:
    \begin{equation}
        \mathcal{E}_C = U^{(t)}\mathcal{E}_{Z_{k_t}}U^{(t-1)}\mathcal{E}_{Z_{k_{t-1}}}\cdots U^{(1)}\mathcal{E}_{Z_{k_1}}U^{(0)}.
    \end{equation}
    By grouping the Clifford gates into cumulative unitaries defined as $W^{(r)} = U^{(r)} U^{(r-1)} \cdots U^{(0)}$, we can rewrite the channel as
    \begin{equation}
        \mathcal{E}_C = 
        W^{(t)}(W^{(t-1)\dagger}\mathcal{E}_{Z_{k_t}}W^{(t-1)})\cdots (W^{(0)\dagger}\mathcal{E}_{Z_{k_1}}W^{(0)}).
    \end{equation}
    
    Conjugating the channel $\mathcal{E}_{Z_k}$ by a Clifford unitary $W$ is equivalent to applying the channel $\mathcal{E}_P$, where $P = W^\dagger Z_k W$ is another Pauli operator.
    By tracking the transformations of the original $Z$ operators through the cumulative Clifford unitaries, we define the Paulis $P^{(r)} = W^{(r-1)\dagger} Z_{k_r} W^{(r-1)}$ to express the entire circuit channel as:
    \begin{equation}
        \mathcal{E}_C = W^{(t)}\mathcal{E}_{P^{(t)}}\mathcal{E}_{P^{(t-1)}}\cdots \mathcal{E}_{P^{(1)}}.
    \end{equation}
    Since conjugating $Z_{k_r}$ by a uniformly random Clifford $W^{(r-1)}$ maps it uniformly across the set of all non-identity Pauli operators with global phases $\pm1$, the $P^{(r)}$ are drawn i.i.d. uniformly from $\pm\mathcal{P}_n \setminus \set{\pm I}$. Note that the purity bound derived in Lemma~\ref{lem:channel-purity} holds identically for $-P$. From the bound obtained in Lemma~\ref{lem:channel-purity}, each channel application dampens the purity such that
    \begin{equation}
        \mathbb{E}_{P \sim \pm\mathcal{P}_n \setminus \set{\pm I}}[\Tr(\mathcal{E}_P(\rho)^2)] < 0.5^n + 0.75(\Tr(\rho^2) - 0.5^n).
    \end{equation}
    Because the expected purity bound obtained in Lemma~\ref{lem:channel-purity} is linear with respect to the input purity $\Tr(\rho^2)$, applying this recursively $t$ times bounds the expected purity of the final state as
    \begin{equation}
        \mathbb{E}_{P_1, P_2, \dots, P_t \overset{i.i.d.}{\sim} \pm\mathcal{P}_n \setminus \set{\pm I}}[\Tr(\rho_b^2)] < 0.5^n + 0.75^t(1 - 0.5^n) < 0.5^n + 0.75^t.
    \end{equation}
    Using the direct relation $\eta^2\le \Tr(\rho_b^2)$ established in Lemma~\ref{lem:purity-bound}, the theorem immediately follows.
\end{proof}

\clearpage
\section{Comparison with Dyadic Channel Simulator} \label{sec:dcs-appendix}
This appendix formalizes the relation between ESSENCE and dyadic channel simulator~(DCS). We first compare the sampling procedures and then prove the near-optimality result for gates satisfying the lifting lemma.

\subsection{Dyadic Channel Simulator}
DCS, which was proposed by Seddon's Ph.D. thesis \cite{DS:2022}, is similar to well-known dyadic frame simulator~(DFS) also made by Seddon \cite{DFS:2020}. DCS decomposes completely positive and trace-preserving~(CPTP) maps dyadically. It utilizes stabilizer-preserving operators~(SPOs), the functions that map every stabilizer states into some stabilizer state, as the fundamental units of updating states.
\begin{equation}
    \text{SPO}_{n,n} = \left\{ K \mid (K \otimes I_n) \ket{\phi} \propto \ket{\phi'} \text{ for some } \ket{\phi'} \in \text{STAB}_{2n}, \text{ for all } \ket{\phi} \in \text{STAB}_{2n} \right\}
\end{equation}

These operators form the building blocks for more complex channel decompositions. Consider a circuit $\mathcal{E} = \mathcal{E}^{(T)}\cdots \mathcal{E}^{(2)}\mathcal{E}^{(1)}$ consisting of a series of channels, where each channel $\mathcal{E}^{(t)}$ has a known decomposition over the set of dyadic stabilizer-preserving maps, denoted by $\text{DSP}_{n,n}$. That is, we can express each channel as $\mathcal{E}^{(t)} = \sum_j \alpha_j^{(t)}\mathcal{T}_j^{(t)}$, where the coefficients $\alpha_j^{(t)} \in \mathbb{C}$ and the gates $\mathcal{T}_j^{(t)} \in \text{DSP}_{n,n}$. The set of dyadic stabilizer-preserving maps is defined as: 
\begin{equation}
    \text{DSP}_{n,n} = \left\{ \rho \mapsto \sum_j L_j \rho R_j^\dagger \ \middle\vert \ L_j, R_j \text{ are SPOs}, \text{ and } \sum_j L_j^\dagger L_j = \sum_j R_j^\dagger R_j = I \right\}.
\end{equation}

To estimate the expectation value
\begin{equation}
    \langle E \rangle = \Tr\left[E\mathcal{E}(|\phi^{(0)}\rangle\langle\phi^{(0)}|)\right]
\end{equation}
of a stabilizer observable $E$ for a state through the circuit from the initial state $\sigma^{(0)} = \ket{\phi^{(0)}}\bra{\phi^{(0)}}$, DCS samples a series of dyadic stabilizer-preserving maps according to its corresponding probability induced by the $\ell_1$-norms of the coefficients $\{\alpha_{j_t}^{(t)}\}$:
\begin{equation}
    \Tr\left[E\mathcal{E}(|\phi^{(0)}\rangle\langle\phi^{(0)}|)\right] = \underbrace{\left(\prod_{t=1}^T\|\boldsymbol{\alpha}^{(t)}\|_1\right)}_{\text{sample range}}\sum_{\{j_t\}}\underbrace{\left(\prod_{t=1}^T\frac{|\alpha_{j_t}^{(t)}|}{\|\boldsymbol{\alpha}^{(t)}\|_1}\right)}_{\text{probability}}\prod_{t=1}^T\frac{\alpha_{j_t}^{(t)}}{|\alpha_{j_t}^{(t)}|}\Tr\left[E\mathcal{T}_{j_T}^{(T)}\cdots\mathcal{T}_{j_1}^{(1)}(\sigma^{(0)})\right]
\end{equation}
The sample range of this process is governed by the $\ell_1$-norms of its decomposition coefficients, defined as the channel's dyadic negativity $\Lambda(\mathcal{E}^{(t)}) = \|\alpha^{(t)}\|_1$, assuming the operator norm $\|E\|$ of the observable is less than or equal to 1. The sub-multiplicativity of dyadic negativity establishes a theoretical lower bound $\Lambda(\mathcal{E}) \le \prod_t \Lambda(\mathcal{E}^{(t)})$ on the range, providing a sufficient sample number of $\mathcal{O}(\Lambda(\mathcal{E})^2)$ for the algorithm's time complexity via Hoeffding's inequality \cite{hoeffding:1963}.

To simulate a map $\mathcal{T} = \sum_j L_j(\cdot)R_j^\dagger$ acting on a density matrix $\sigma = \ket{\phi_L}\bra{\phi_R}$, we must pick each component $L_j(\sigma)R_j^\dagger$ with probability $p_j$ below, where $\|A\|_1=Tr(\sqrt{AA^\dagger})$ is the trace norm,
\begin{equation}
    p_j = \|L_j\ket{\phi_L}\bra{\phi_R}R_j^\dagger\|_1, 
\end{equation}
and returns the zero operator with probability
\begin{equation}
    p_0 = 1 - \sum_{j=1}^{N}{p_j}.
\end{equation}

Since $L_j$ and $R_j$ are not unitary, probabilities $p_j$ depend on $\sigma$ and cannot be pre-calculated. Therefore, these are calculated dynamically, by updating $\ket{\phi_L} \mapsto L\ket{\phi_L}$ and $\ket{\phi_R} \mapsto R\ket{\phi_R}$ for selected $L$ and $R$ which is efficiently calculated with CH-form representation.

Although the DCS framework provides an exceptional degree of generality by encompassing the foundations of both stabilizer propagation \cite{SP:2017} and DFS, its implementation is caught in a fundamental tension between theoretical reach and computational tractability. This tension is most immediately apparent in the scaling of the basis set. As the framework expands to capture more complex maps, finding an optimal decomposition rapidly becomes intractable. Even when restricting the search space from general DSP to projected dyadic stabilizer-preserving maps, a single qubit yields an overwhelming 93,312 extreme points. This creates a harsh trade-off: while a broader optimization space can theoretically drive down dyadic negativity, it simultaneously pushes the search for a minimal decomposition beyond the limits of practical computation.

Compounding this difficulty is the overhead of real-time probability estimation. While calculating individual transition probabilities $p_j$ requires only polynomial time, the overall utility of the simulation hinges on the total number of decomposition terms remaining within $\poly(n)$. If the decomposition size scales exponentially, any theoretical gains achieved through a lower dyadic negativity are entirely erased by the ensuing computational explosion.

Ultimately, while DCS stands as a mathematically elegant framework unifies prior simulation methods, its practical utility remains severely constrained by the dual challenges of basis optimization and real-time execution overhead.

\subsection{ESSENCE as a Specialization of DCS} \label{sec:dcs-equivalence}
For arbitrary unitary gate $U$, denote $\mathcal{U}$ as a unitary channel defined by $\mathcal{U}(\rho)=U \rho U^\dagger$. From the optimal decomposition of $U=\sum_i{\alpha_iK_i}$ with Clifford gates in ESSENCE, we can rewrite the channel $\mathcal{U}$ as
\begin{equation}\label{eqn:xi-and-lambda}
    \mathcal{U}(\cdot)=\sum_{i,j}\alpha_i\alpha_j^\star K_i(\cdot)K_j^\dagger,
\end{equation}
which is a valid dyadic channel decomposition as $K_i(\cdot)K_j^\dagger \in DSP_{n,n}$ for all $i, j$. Its $\ell_1$-norm of the coefficients is 
\begin{equation}
    \sum_{i,j}{|\alpha_i\alpha_j^\star|}=(\sum_i|\alpha_i|)^2 =\xi(U),
\end{equation}
despite the fact that it may not be optimal in dyadic sense.
Thus, the decomposition induced by SoC proves that $\xi(U)\ge \Lambda(\mathcal{U})$.

Regarding the sampling procedure of ESSENCE, the sample and propagate algorithm of ESSENCE is equivalent with the process of DCS with specific decomposition over $\text{DSP}_{n,n}$, as summarized in Table~\ref{tab:dcs-specialization}. It samples two different Cliffords $K_i, K_j$ with probability $|\alpha_i|/\|\alpha\|_1$ and $|\alpha_j|/\|\alpha\|_1$ independently and update two stabilizer states $\ket{\phi_1}$ and $\ket{\phi_2}$, respectively. The algorithm of the DCS can run the same procedure by sampling $K_i(\cdot)K_j^\dagger \in DSP_{n,n}$ with probability $|\alpha_i\alpha_j|/\|\alpha\|_1^2$ and updating dyadic frame $\sigma = \ket{\phi_1}\bra{\phi_2}$ accordingly.

\begin{table}[ht]
    \centering
    \caption{Sampling-step correspondence between ESSENCE and the dyadic channel simulator~(DCS).}
    \label{tab:dcs-specialization}
    \vspace{0.5em}
    \begin{tabular}{lll}
        \toprule
        Component & ESSENCE & DCS equivalent \\
        \midrule
        Samples & $K_i, K_j$ & $K_i(\cdot)K_j^\dagger$ \\
        Sample probability & $|\alpha_i|/\|\boldsymbol{\alpha}\|_1$, $|\alpha_j|/\|\boldsymbol{\alpha}\|_1$ & $|\alpha_i\alpha_j^*|/\|\boldsymbol{\alpha}\|_1^2$ \\
        Updates & $\ket{\phi_1} \mapsto K_i\ket{\phi_1}, \ket{\phi_2} \mapsto K_j\ket{\phi_2}$ & $\ket{\phi_1}\bra{\phi_2} \mapsto K_i\ket{\phi_1}\bra{\phi_2}K_j^\dagger$ \\
        \bottomrule
    \end{tabular}
\end{table}

Thus, ESSENCE can be seen as a practical specialization of DCS. It leverages SoC to reduce the search size to decompose, and also to simplify the calculation of transition probability. The trade-off of this specialization is replacing the range parameter $\Lambda(\mathcal{U})$ to $\xi(U)$ which is always larger than or equal to original. We found a sufficient condition, that if a gate satisfies the lifting lemma from \cite{bravyi:2019} then the equality holds.

\begin{restatable}[Near Optimality]{theorem}{nearoptimalThm} \label{thm:near-optimality}
    Suppose $V$ satisfies the lifting lemma. Then, $\xi(V) = \Lambda(\mathcal{V})$.
\end{restatable}
\begin{corollary}
    Let a $t$-qubit diagonal unitary gate $V$ satisfy the lifting lemma. Then, the simulation of $V$ using ESSENCE is equivalent using DCS with an optimal decomposition of $V$ over $\text{DSP}_{n,n}$.
\end{corollary}
For example, $R_Z(\theta)$ and $CCZ$ gates satisfy the lifting lemma, or the conditions (1) being diagonal, (2) their equatorial stabilizer decompositions are known to be optimal. In such cases, ESSENCE is an implementation of DCS, with exact and simple decomposition over $\text{DSP}_{n,n}$. Therefore, it is a practical realization of the pure abstract DCS making the algorithm computable in real life.

\subsection{Near-optimality Theorem}
This section proves that satisfying the condition of the lifting lemma \cite{bravyi:2019} implies the equality $\xi(V) = \Lambda(\mathcal{V})$. We achieve this by establishing two key lower bounds (Lemma~\ref{lemma:choi-negativity} and Lemma~\ref{lemma:magic-state}), which are then used to construct a chain of inequalities that ultimately collapses.

\subsubsection{Prerequisite: The Lifting Lemma}
\begin{lemma} [Lifting lemma \cite{bravyi:2019}]
    Suppose $V$ is a $t$-qubit diagonal unitary gate.
    Suppose $V\ket{+^{\otimes t}}=\ket{V}=\sum_j{\alpha_j\ket{\phi_j}}$ where $\ket{\phi_j}$ are equatorial stabilizer states and this is an optimal decomposition for magic state. Then, $\xi(V) = \xi(\ket{V})$.
\end{lemma}

\subsubsection{Supporting Lemmas}
Before proving the main theorem, we must establish the relationship between the negativity of a channel and its Choi-state, as well as the stabilizer extent of the magic state.

\begin{lemma}\label{lemma:choi-negativity}
    For any arbitrary CPTP channel $\mathcal{V}$, the dyadic negativity of the channel is lower bounded by the dyadic negativity of its Choi-state $\tilde{\mathcal{V}}$, or
    \begin{equation}
        \Lambda(\mathcal{V}) \ge \Lambda(\tilde{\mathcal{V}}).
    \end{equation}
\end{lemma}
\begin{proof}
    We show that under the transformation of a channel to a Choi-state, any $\text{DSP}_{n,n}$ map gives a valid stabilizer dyad decomposition with an $\ell_1$-norm less than or equal to 1. This technique was proposed in \cite{DS:2022} but we will introduce it again.

    Let $\mathcal{T} \in DSP_{n,n}$ and its Choi-state be $\tilde{\mathcal{T}}=\mathcal{T}\otimes \mathcal{I}(\ket{\Phi^+}\bra{\Phi^+})$.
    $\mathcal{T}$ is defined by pair of Kraus operators $(L_j, R_j)$.
    Thus, $\tilde{\mathcal{T}}=\sum_j{(L_j\otimes I) \ket{\Phi^+}\bra{\Phi^+}(R_j^\dagger\otimes I)}$. By the definition of $\text{SPO}_{n,n}$, the stabilizer states $(L_j\otimes I) \ket{\Phi^+}, (R_j \otimes I)\ket{\Phi^+}$ need not be normalized.
    With normalization into $\ket{L_j}$ and $\ket{R_j}$, respectively, it gives a dyadic stabilizer decomposition:
    \begin{equation}
        \tilde{\mathcal{T}}=\sum_j{v_jw_j\ket{L_j}\bra{R_j}}
    \end{equation}
    where $v_j = \|(L_j \otimes I)\ket{\Phi^+}\|$ and $w_j = \|(R_j \otimes I)\ket{\Phi^+}\|$.
    Using the completeness of Kraus operators ($\sum_j{L_j^\dagger L_j}=\sum_j{R_j^\dagger R_j} = I$), we can show that $\sum_j{v_j^2}=\sum_j{w_j^2}=1$:
    \begin{equation}
        \begin{split}
            \sum_j{v_j^2} &= \sum_j{\|(L_j \otimes I)\ket{\Phi^+}\|^2} = \sum_j{\bra{\Phi^+}(L_j^\dagger \otimes I)(L_j \otimes I)\ket{\Phi^+}} = \bra{\Phi^+}I\ket{\Phi^+}=1\\
            \sum_j{w_j^2} &= \sum_j{\|(R_j \otimes I)\ket{\Phi^+}\|^2} = \sum_j{\bra{\Phi^+}(R_j^\dagger \otimes I)(R_j \otimes I)\ket{\Phi^+}} = \bra{\Phi^+}I\ket{\Phi^+}=1
        \end{split}
    \end{equation}
    By the triangle inequality, $\tilde{\mathcal{T}}$ yields a natural dyadic decomposition with $\ell_1$-norm at most 1:
    \begin{equation}
        \sum_j|v_jw_j| \le \sqrt{(\sum_jv_j^2)(\sum_jw_j^2)} = 1.
    \end{equation}
    
    Let $\mathcal{V}=\sum_k{\alpha_k\mathcal{T}_k}$ be the optimal channel decomposition of $\mathcal{V}$, then its Choi-state $\tilde{\mathcal{V}} = \sum_k{\alpha_k\tilde{\mathcal{T}_k}}$ gives a dyadic decomposition with its $\ell_1$-norm at most $\sum_k{\abs{\alpha_k}} = \Lambda(\mathcal{V})$. Therefore, by the definition of the negativity,
    \begin{equation}
        \Lambda(\mathcal{V}) \ge \Lambda(\tilde{\mathcal{V}}).
    \end{equation}
\end{proof}

\begin{lemma}\label{lemma:magic-state}
    Suppose $V$ is a diagonal unitary operator and $\ket{\Phi_V}=(V_A \otimes I_B)\ket{\Phi^+}$. Then, the stabilizer extent satisfies:
    \begin{equation}
        \xi(\ket{\Phi_V}) \ge \xi(\ket{V}).
    \end{equation}
\end{lemma}
\begin{proof}
    As $V$ is a diagonal operator, we can let
    \begin{equation}
        V=\sum_{x=0}^{2^{n-1}}{e^{i\theta_x}\ket{x}\bra{x}}
    \end{equation}
    Applying this to the maximally entangled state yields,
    \begin{equation}
        \ket{\Phi_V}=2^{-n/2}\sum_{x=0}^{2^n-1} {(V_A \otimes I_B)\ket{x}_A\ket{x}_B}
        = 2^{-n/2}\sum_{x=0}^{2^n-1}{e^{i\theta_x}\ket{x}_A\ket{x}_B}
    \end{equation}
    Let $U=\bigotimes_{k=1}^{n}{\text{CNOT}_{A_k,B_k}}$. $U$ acts as $U\ket{x}_A\ket{y}_B=\ket{x}_A\ket{x\oplus y}_B$. Then,
    \begin{equation}
        U\ket{\Phi_V}=2^{-n/2}\sum_{x=0}^{2^n-1}{e^{i\theta_x}U\ket{x}_A\ket{x}_B}
        = 2^{-n/2}\sum_{x=0}^{2^n-1}{e^{i\theta_x}\ket{x}_A\ket{0^n}_B}
        = \ket{V}_A\ket{0^n}_B
    \end{equation}
    Since applying Clifford operations $U$ and tensoring a stabilizer ancilla do not increase stabilizer extent, we obtain the lower bound 
    \begin{equation}
        \xi(\ket{\Phi_V}) \ge \xi(\ket{V}).
    \end{equation}
\end{proof}
\subsubsection{Near-optimality Theorem}
With the necessary lemmas established, we are now ready to prove our main result.
\nearoptimalThm*
\begin{proof}
    We establish the identity $\xi(V) = \Lambda(\mathcal{V})$ by constructing a chain of inequalities:
    \begin{equation}
        \xi(V) \ge \Lambda(\mathcal{V}) \ge \Lambda(\tilde{\mathcal{V}}) \ge \xi(\ket{\Phi_V}) \ge \xi(\ket{V}),
    \end{equation}
    The Choi state $\tilde{\mathcal{V}}$ of $\mathcal{V}$ is defined as $\tilde{\mathcal{V}} = (\mathcal{V}_A \otimes \mathcal{I}_B)(\ket{\Phi^+}\bra{\Phi^+})$, where the maximally entangled state is $\ket{\Phi^+}=2^{-n/2}\sum_{x=0}^{2^n-1}\ket{x}_A\ket{x}_B$ and $\ket{\Phi_V}=(V_A \otimes I_B)\ket{\Phi^+}$. As $V$ is unitary, $\tilde{\mathcal{V}}=\ket{\Phi_V}\bra{\Phi_V}$ is a pure state. 
    
    We justify each step in the chain as follows:
    \begin{enumerate}
        \item $\xi(V) \ge \Lambda(\mathcal{V})$: The equivalence of stabilizer decompositions presented in Eq.~\eqref{eqn:xi-and-lambda} ensures that the stabilizer extent of the operator $V$ upper bounds the negativity of its associated channel $\mathcal{V}$.
        \item $\Lambda(\mathcal{V}) \ge \Lambda(\tilde{\mathcal{V}})$: This lower bound is established in Lemma~\ref{lemma:choi-negativity}.
        \item $\Lambda(\tilde{\mathcal{V}}) = \xi(\ket{\Phi_V})$: Since $\tilde{\mathcal{V}}$ is a pure state, its negativity corresponds exactly to the stabilizer extent by Lemma 1 in \cite{DFS:2020}.
        \item $\xi(\ket{\Phi_V}) \ge \xi(\ket{V})$: This lower bound is established in Lemma~\ref{lemma:magic-state}.
    \end{enumerate}
    Finally, if the lifting lemma holds, we know that $\xi(\ket{V}) = \xi(V)$. This forces the entire chain to collapse into a series of equalities, directly proving that:
    \begin{equation}
        \xi(V) = \Lambda(\mathcal{V})
    \end{equation}
\end{proof}

\clearpage
\section{Additional Simulation Results}\label{sec:results-appendix}
This appendix reports supplementary simulation data for Section~\ref{sec:results}.

\subsection{Additional Results for QAOA Ansatz} \label{sec:results-appendix-qaoa}
 This section provides further details on the N20D3 and N60D4 benchmark from Section~\ref{sec:results-qaoa}.

\begin{figure}[htbp]
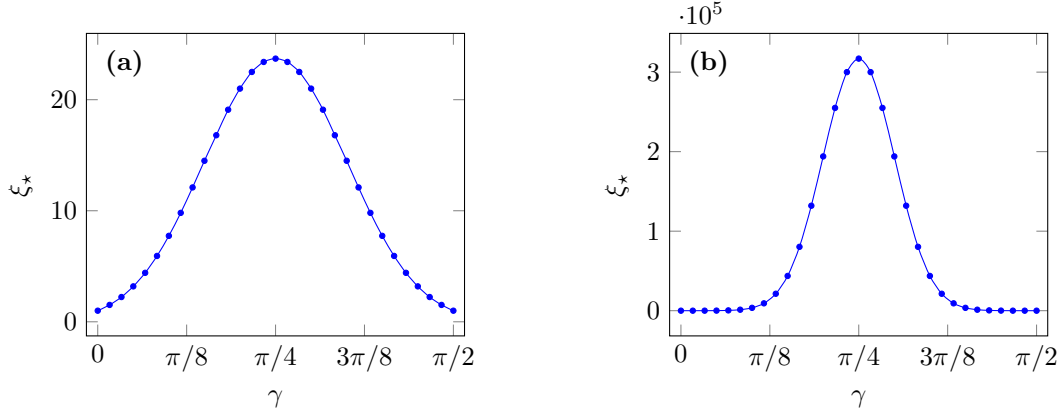

    \centering
    \begin{subfigure}{.48\textwidth}
        \centering
        \phantomsubcaption
        \label{fig:qaoa-extra-n20-xi}
        \qaoaXi{data/details_N20D3.csv}{5}{4}{(a)}
    \end{subfigure}
    \hfill
    \begin{subfigure}{.48\textwidth}
        \centering
        \phantomsubcaption
        \label{fig:qaoa-extra-n60-xi}
        \qaoaXi{data/details_N60D4.csv}{5}{4}{(b)}      
    \end{subfigure}    
    \caption{
    Gate-wise stabilizer extent $\xi_\star$ as a function of $\gamma$ for (a) the N20D3 instance and (b) the N60D4 instance.}
    \label{fig:qaoa-extra}
\end{figure}

\begin{figure}[H]
    \centering
    \graphscatterC{data/details_N20D3.csv}{5}{4}{Sample counts $h_{\mathrm{Hoeffding}}$ and $h_\mathrm{CARVE}$ in the benchmark of N20D3, as functions of $\xi_\star$ on logarithmic axes. Dashed segments indicate the reference scalings $h\propto\xi_\star$ and $h\propto\xi_\star^2$.}
    \label{fig:qaoa-extra-n20-scatter}
\end{figure}

\subsection{Additional Results for Quantum Kernel Estimation} \label{sec:results-appendix-qke}

This section provides further details on the $n=2$ and $n=64$ benchmark from Section~\ref{sec:results-qsvm} alongside a newly simulated $n=5$ instance. All simulations employ the same settings as the main text: the Pauli feature map $U(x)$, $\epsilon=0.2$, and $\delta=0.2$.

\begin{figure}[htbp]
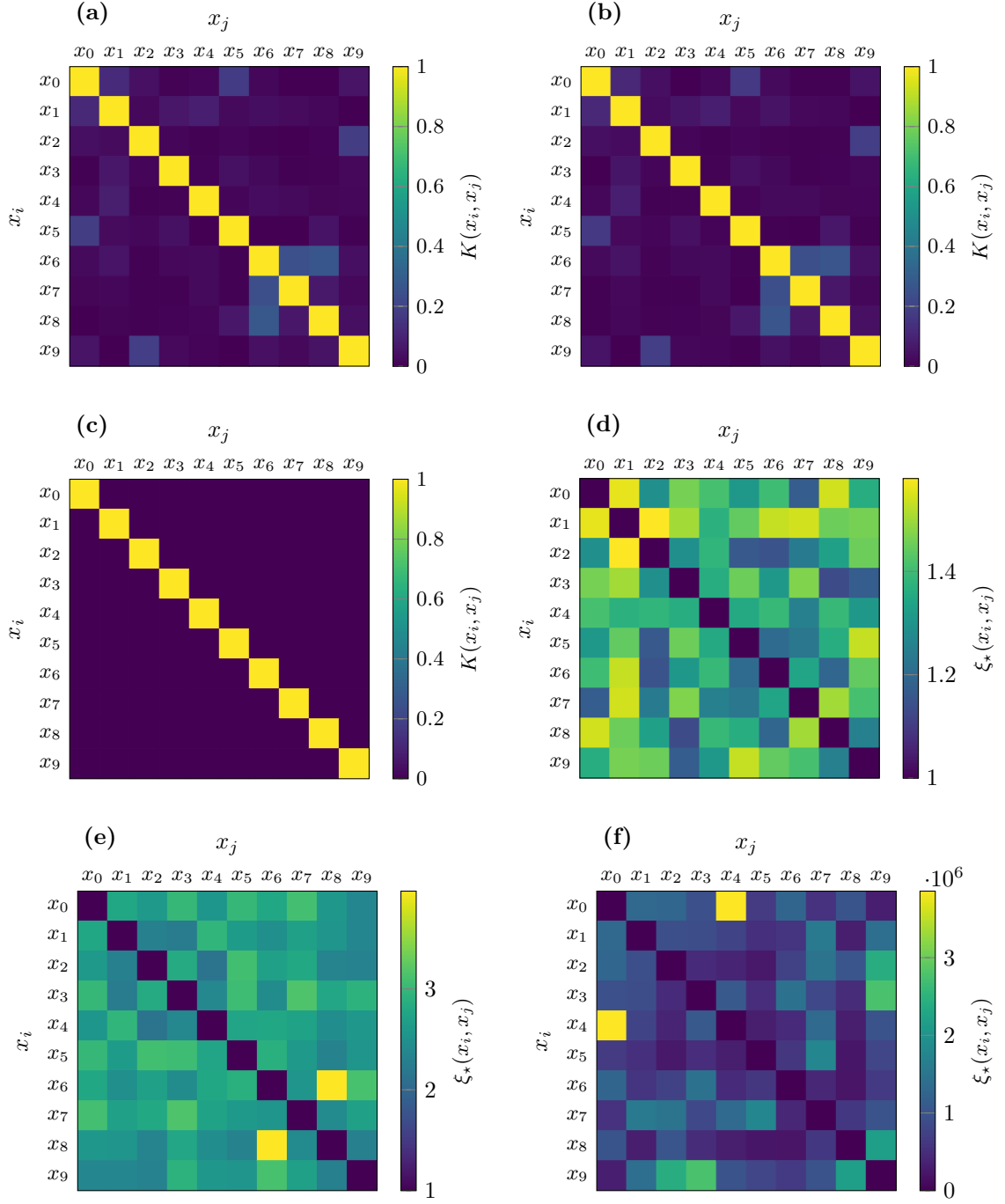

    \centering
    \begin{subfigure}{.48\textwidth}
        \centering        
        \phantomsubcaption
        \label{fig:qke-heatmap-n5-res}
        \qkeResult{data/details_QKE_n5.csv}{4.5}{4.5}{(a)}
    \end{subfigure}
    \hfill
    \begin{subfigure}{.48\textwidth}
        \centering
        \phantomsubcaption
        \label{fig:qke-heatmap-n5-truth}
        \qkeTruth{data/details_QKE_n5.csv}{4.5}{4.5}{(b)}
    \end{subfigure}
    \vspace{1em}

    \begin{subfigure}{.48\textwidth}
        \centering        
        \phantomsubcaption
        \label{fig:qke-heatmap-n64-res}
        \qkeResult{data/details_QKE_n64.csv}{4.5}{4.5}{(c)}
    \end{subfigure}
    \hfill
    \begin{subfigure}{.48\textwidth}
        \centering
        \phantomsubcaption
        \label{fig:qke-heatmap-n2-xi}
        \qkeXi{data/details_QKE_n2.csv}{4.5}{4.5}{(d)}
    \end{subfigure}
    \vspace{1em}
    
    \begin{subfigure}{.48\textwidth}
        \centering
        \phantomsubcaption
        \label{fig:qke-heatmap-n5-xi}
        \qkeXi{data/details_QKE_n5.csv}{4.5}{4.5}{(e)}
    \end{subfigure}
    \hfill
    \begin{subfigure}{.48\textwidth}
        \centering
        \phantomsubcaption
        \label{fig:qke-heatmap-n64-xi}
        \qkeXi{data/details_QKE_n64.csv}{4.5}{4.5}{(f)}
    \end{subfigure}    
    \caption{
    Quantum kernel estimation results for $n = 2, 5$ and $64$. (a) Estimated and (b) exact kernel matrices $K(x_i,x_j)$ for $n=5$. (c) Estimated kernel matrix $K(x_i,x_j)$ for $n=64$. (d)-(f) Gate-wise stabilizer extent $\xi_\star(x_i, x_j)$ matrices for $n=2, 5$, and $64$, respectively.}
    \label{fig:qke-extra}
\end{figure}
\clearpage

\begin{figure}[htbp]
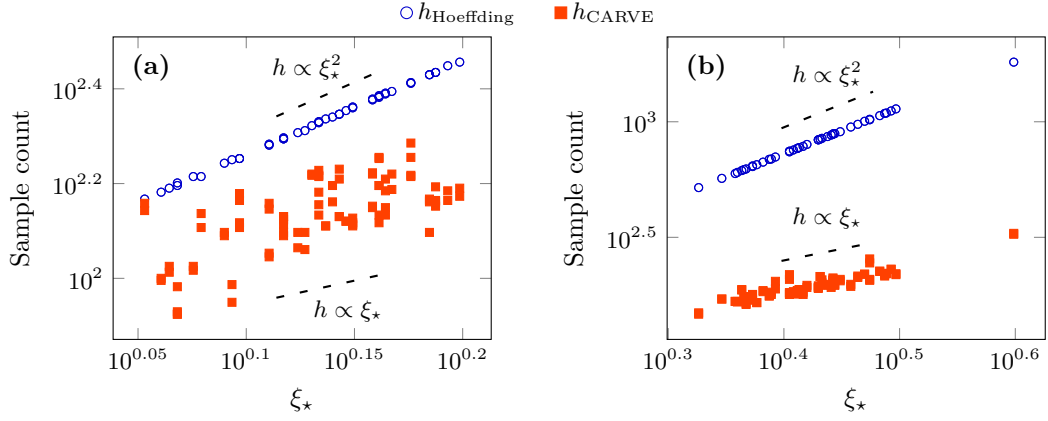

    \centering
    \ref*{sharedlegend2}
    \vspace{0.0cm}
    
    \begin{subfigure}{.48\textwidth}
        \centering
        \phantomsubcaption
        \graphscatterD{data/details_QKE_n2.csv}{5}{4}
        \label{fig:qke-extra-n2-scatter}
    \end{subfigure}
    \begin{subfigure}{.48\textwidth}
        \centering
        \phantomsubcaption
        \graphscatterE{data/details_QKE_n5.csv}{5}{4}
        \label{fig:qke-extra-n5-scatter}
    \end{subfigure}
    \caption{Additional sample counts $h_{\mathrm{Hoeffding}}$ and $h_\mathrm{CARVE}$ as functions of $\xi_\star$ on logarithmic axes. Dashed segments indicate the reference scalings $h\propto\xi_\star$ and $h\propto\xi_\star^2$. Quantum kernel estimation benchmark results for (a) $n=2$ and (b) $n=5$.}
    \label{fig:qke-extra-scatter}
\end{figure}

Figure~\ref{fig:qke-extra} illustrates both the estimated kernel matrices and their corresponding gate-wise stabilizer extents. For $n=5$, comparing the estimated heatmaps with the exact state-vector ground truth visually confirms the accuracy of our approach, with entrywise absolute errors remaining well below $\epsilon = 0.2$. For $n=64$, the kernel heatmap closely resembles an identity matrix. This behavior is numerically anticipated, as independently generated high-dimensional feature states become nearly orthogonal, rendering the off-diagonal entries negligible. 

\end{document}